\documentclass[10pt]{amsart} 

\usepackage{fullpage}

\usepackage{url}

\usepackage{amssymb}

\usepackage{cite}

\usepackage{graphicx}

\usepackage[usenames]{color}

\usepackage{accents}

\usepackage{hyperref}

\renewcommand{\a}{\alpha}
\renewcommand{\b}{\beta}
\newcommand{\g}{\gamma}
\renewcommand{\d}{\delta}

\newcommand{\e}{\varepsilon}
\newcommand{\f}{\varphi}
\newcommand{\s}{\sigma}
\newcommand{\Si}{\Sigma}
\renewcommand{\k}{\kappa}
\renewcommand{\l}{\lambda}

\newcommand{\cO}{{\mathcal O}}

\newcommand{\cC}{{\mathcal C}}

\newcommand{\cT}{{\mathcal T}}

\newcommand{\cB}{{\mathcal B}}
\newcommand{\cL}{{\mathcal L}}
\newcommand{\cE}{{\mathcal E}}
\newcommand{\cU}{{\mathcal U}}

\newcommand{\cN}{{\mathcal N}}

\newcommand{\cH}{{\mathcal H}}

\newcommand{\cI}{\mathcal I}

\newcommand{\bR}{\mathbb R}

\newcommand{\bE}{\mathbb E}

\newcommand{\bV}{\mathbb V}
\newcommand{\Diff}{\mathrm{Diff}}
\newcommand{\Ric}{\mathrm{Ric}}

\newcommand{\tT}{\mathsf{T}}

\newcommand{\be}{\begin{equation}}
\newcommand{\ee}{\end{equation}}
\newcommand{\bes}{\begin{equation*}}
\newcommand{\ees}{\end{equation*}}

\newcommand{\tr}{\mathrm{tr}}

\newcommand{\beaa}{\begin{eqnarray*}}
\newcommand{\bea}{\begin{eqnarray}}
\newcommand{\beal}[1]{\begin{eqnarray}\label{#1}}
\newcommand{\bean}{\begin{eqnarray}\nonumber}
\newcommand{\beadl}[1]{\begin{deqarr}\label{#1}}
\newcommand{\eeadl}[1]{\arrlabel{#1}\end{deqarr}}
\newcommand{\eeal}[1]{\label{#1}\end{eqnarray}}
\newcommand{\eead}[1]{\end{deqarr}}
\newcommand{\eea}{\end{eqnarray}}
\newcommand{\eeaa}{\end{eqnarray*}}

\newcommand{\p}{\partial}

\renewcommand{\to}{\rightarrow}

\renewcommand{\phi}{\varphi}
\renewcommand{\epsilon}{\varepsilon}
\renewcommand{\hat}{\widehat}

\newcommand{\<}{\langle}
\renewcommand{\>}{\rangle}

\newcommand{\w}{\widetilde}

\theoremstyle{plain}
\newtheorem{lemma}{Lemma}[section]
\newtheorem{proposition}[lemma]{Proposition}
\newtheorem{theorem}[lemma]{Theorem}

\newtheorem{corollary}[lemma]{Corollary}

\theoremstyle{remark}
\theoremstyle{definition}
\newtheorem{remark}[lemma]{Remark}

\def\blacksquare{\hbox to .60em {\vrule width .60em height .60em}}

\numberwithin{equation}{section}

\begin{document}

\title[ ]{Well-Posed Geometric Boundary data in General Relativity, II: Twisted Dirichlet boundary data} 

\author{Zhongshan An and Michael T. Anderson}

\address{Institute of Geometry and Physics, 
University of Science and Technology of China,
No. 99 Xiupu Road, Shanghai, China, 201305}
\email{zhshan.an@gmail.com}

\address{Department of Mathematics, 
Stony Brook University,
Stony Brook, NY 11794}
\email{michael.anderson@stonybrook.edu}

\begin{abstract} 
In this second work in a series, we prove the local-in-time well-posedness of the IBVP for the vacuum Einstein equations in general 
relativity with twisted Dirichlet boundary conditions on a finite timelike boundary. The boundary conditions consist of specification of the 
pointwise conformal class of the boundary metric, together with a scalar density involving a combination of the volume form of the bulk metric 
restricted to the boundary together with the volume form of the boundary metric itself. 
\end{abstract}

\maketitle 

\section{Introduction}

  This work is the second in a series on the initial boundary value problem (IBVP) in general relativity. The setting is the study of globally hyperbolic vacuum 
solutions of the Einstein equations 
\be \label{vacuum}
\Ric_g = 0,
\ee
on a finite domain $M \simeq I\times S$, where the Cauchy surface $S$ is a compact, $n$-dimensional, ($n \geq 2$), manifold with 
boundary $\p S = \Si$ and $\cC = I\times \Si$ is a time-like boundary $\cC$ of $M$. Here globally hyperbolic is understood in the sense of 
manifolds with time-like boundary. 

  Consider the space $\bE$ of all such vacuum solutions with a fixed topology $M$. Let $\Diff_0(M)$ 
be the group of diffeomorphisms $\f: M \to M$, isotopic to the identity and equal to the identity on $S \cup \cC$. The group $\Diff_0(M)$ acts naturally on 
$\bE$ and the quotient space 
\be \label{cE}
\cE = \bE/{\rm Diff}_0(M)
\ee
is the moduli space of vacuum solutions, closely related to the covariant phase space of vacuum Einstein metrics. For simplicity, all data here and 
below are assumed to be $C^{\infty}$ smooth; cf.~Remark \ref{findiff} for extension to function spaces of finite differentiability. 

\medskip 

  The primary issue here is to find effective descriptions or parametrizations of $\cE$ in terms of initial data on $S$ and boundary data on $\cC$. Let then 
$\cI_0$ denote the space of vacuum initial data, i.e.~pairs $(\g_S, \k)$ consisting of a Riemannian metric $\g_S$ on $S$ and a symmetric bilinear form $\k$ on $S$, 
satisfying the vacuum Einstein constraint equations, (discussed in detail later). In \cite{I}, we considered Dirichlet boundary data, where the space of boundary data is 
the space of globally hyperbolic Lorentz metrics on $\cC$. In this work, we choose the space of boundary data 
\be \label{bspace}
\cB = {\rm Conf}(\cC)\times C_+^{\infty}(\cC),
\ee
to be the space of pointwise conformal classes of globally hyperbolic Lorentz metrics $\g_{\cC}$ on $\cC$ times the space of smooth positive 
scalar functions on $\cC$. The boundary data of $(M, g)$ are now given by the pair 
\be \label{bdata}
([g_{\cC}], \mu_g),
\ee
consisting of the conformal class $[g_{\cC}]$ of the induced boundary metric, together with the scalar 
\be \label{mug}
\mu_g = dv_g (dv_{g_{\cC}})^{-2/n}.
\ee
Here $dv_g$ is the volume form of the ambient metric $g$ on $M$ restricted to $\cC$ and $dv_{g_{\cC}}$ is the volume form of the boundary metric 
$g_{\cC}$ on $\cC$. The product $\mu_g$ is viewed as a scalar density with respect to an arbitrarily chosen background volume form on $M$. 
The particular form of $\mu_g$ in \eqref{mug} is chosen since it is well adapted to the choice of gauge used here - the harmonic or wave coordinate 
gauge. Different choices of gauge may lead to other expressions for the scalar $\mu_g$, but this will not be pursued further here. 

Note that the pair of boundary data $([g_{\cC}], dv_{g_{\cC}})$, when taken together, is exactly Dirichlet boundary data $g_{\cC}$ studied in \cite{I}. While 
Dirichlet boundary data was proved to form a well-posed IBVP in \cite{I} in some regimes it is not well-posed in general; as discussed further brlow, some 
restrictions are necessary. The boundary data \eqref{bdata} can be viewed as a modification of Dirichlet boundary data, involving a twisting by the volume 
form of $g$. Observe that the boundary data \eqref{bdata} involves no derivatives of the metric at $\cC$. 

\medskip 

For $\cB$ as in \eqref{bspace}, there is a natural map to initial and boundary data given by 
\be \label{mainPhi}
\begin{split}
&\Phi: \bE \to \cI_0 \times_c \cB,\\
&\Phi(g) = (g_S, K, [g_{\cC}], \mu_g),\\
\end{split}
\ee
where $g_S$ is the metric on $S$ induced by $g$ and $K$ is the induced second fundamental form (extrinsic curvature) on $S$. 
The subscript $c$ denotes the $C^{\infty}$ compatibility conditions between the initial and boundary data at the corner $\Si$. 

  Due to the diffeomorphism invariance of the vacuum Einstein equations, it is well-known that the map $\Phi$ cannot be one-to-one and so $\Phi$ 
cannot be an effective parametrization of the space of vacuum solutions $\bE$. Let $\bE^H$ be the space of vacuum Einstein metrics in harmonic or 
wave gauge; cf.~\S 2.4 for details. Let $\Diff_1(M) \subset \Diff_0(M)$ be the subgroup of diffeomorphisms $\f$ such that $\f$ is the identity to first order at $S$; 
equivalently, $\f \in \Diff_0(M)$ and 
$$\f_*(\p_t) = \p_t \ \ {\rm on} \ \ S.$$
There are no $1^{\rm st}$ order restrictions for $\f \in \Diff_0(M)$ along $\cC$, except for the smooth compatibility condition at the corner $\Si$. For any 
metric $g$ such diffeomorphisms leave the future time-like normal vector $\nu_S$ to $S$ in $(M, g)$ invariant: $\f_*(\nu_S) = \nu_S$. 
It is well-known (cf.~again \S 2.4)  that for any $g \in \bE$, there is a unique diffeomorphism $\f \in \Diff_1(M)$ such that $$\f^*g \in \bE^H.$$
In other words, there is a natural identification 
\be \label{harmslice}
\cE_* =  \bE/\Diff_1(M) \simeq \bE^H,
\ee
where $\cE_*$ is the marked or framed moduli space, consisting of the moduli space $\cE$ together with the specification of the unit normal $\nu_S$ 
along $S$. The group $\Diff_0(M)$ acts on $\cE_*$ with quotient the moduli space $\cE$,
$$\cE_*/\Diff_0(M) = \cE,$$
and with fibers given by the orbits of the action of the quotient group $\Diff_0(M) / \Diff_1(M)$ on $\cE_*$.\footnote{It is easy to see that $\Diff_1(M)$ is a 
normal subgroup of $\Diff_0(M)$. }

\medskip 

   The main result of this paper is the following (local-in-time) well-posedness result for the IBVP with boundary data given by \eqref{bdata}. 
Let $\bV_S'$ be the space of smooth vector fields along $S$ nowhere tangent to $S$ (the target space for normal vector fields along $S$) and form the 
product space $\cI_0 \times \bV' \times \cB$; let $[\cI_0 \times \bV'_S \times \cB]_c \subset \cI_0 \times \bV' \times \cB$ be the subset consisting of 
data smoothly compatible at the corner $\Si$, cf.~\S 2.4 for details.

 Let $(\bE^H)^* \subset \bE^H$ be the space of vacuum Einstein metrics in harmonic gauge defined (at least) on the domain $[0, t^*)\times S$, 
 where 
 \be \label{tstar}
 t^*: \bE^H \to \bR
 \ee
 is a continuous function of $g$ representing the $g$-proper time between the slices $\{0\} \times S$ and $\{t^*\}\times S$. 
 Define the target data space $[\cI_0 \times \bV'_S \times \cB]_c^*$ in the same way, so that the data on $\cC$ are restricted to $[0, t^*)\times \Si$.

\begin{theorem} \label{mainthm}
There exists a choice of smooth time function $t^*$ such that the spaces $(\bE^H)^*$ and $[\cI_0 \times \bV'_S \times \cB]_c^*$ are smooth tame 
Frechet manifolds and the map 
\be \label{mainPhiH}
\begin{split}
&\Phi^H: (\bE^H)^* \to [\cI_0 \times \bV'_S \times \cB]_c^*,\\
&\Phi^H(g) = (g_S, K, \nu_S, [g_{\cC}], \mu_g),\\
\end{split}
\ee
is a smooth tame diffeomorphism of $(\bE^H)^*$ onto $[\cI_0 \times \bV'_S \times \cB]_c^*$ with smooth, tame inverse. Thus the IBVP with 
boundary data \eqref{bdata} is locally-in-time well-posed in harmonic gauge. 

\end{theorem}

  Briefly, smooth vacuum Einstein metrics on $M$, with arbitrarily given smoothly compatible initial and boundary data, exist for a short time, i.e.~in a 
neighborhood $[0, t^*)\times S$ of $S$. Such solutions are uniquely determined in harmonic gauge by their initial data in 
$\cI_0 \times \bV_S'$ and boundary data in $\cB$. Further, such solutions and a time-of-existence function $t^*$ depend smoothly on the 
target data in $[\cI_0 \times \bV'_S \times \cB]_c$; the IBVP is stable with respect to variation of both the initial data and the 
boundary data in $\cB$. 

A version of Theorem \ref{mainthm} holds also off-shell (for general metrics), cf.~Theorem \ref{offshell}. We refer to \cite{Ham} for details 
regarding Frechet spaces and tame maps. 

   The work and discussion to follow is local-in-time and so we we drop extra symbol $^*$ as in \eqref{mainPhiH} for notational simplicity, unless otherwise indicated. 
   
  We note here a simple but significant consequence of the proof of Theorem \ref{mainthm}:
  
  \begin{corollary} \label{Eman}
  The marked moduli space $\cE_*$ and the moduli space $\cE$ of vacuum Einstein metrics on $M$ as above are smooth tame Frechet manifolds. 
  \end{corollary} 
    
    Note that Corollary \ref{Eman} does not depend on any choice of boundary data space $\cB$; it is intrinsic to the structure of the (local-in-time) 
moduli space $\cE$, independent of any map to boundary data. The smoothness of the moduli space $\cE$ has not been established previously. In fact it is known 
to be false in general for the pure Cauchy problem with a compact Cauchy surface $S$ having empty boundary, $\p S = \emptyset$; this issue is closely related 
to the failure of the linearization stability for solutions of the Cauchy problem at Killing initial data, cf.~\cite{Mon}. 

\medskip 

  The boundary data $([g_{\cC}], \mu_g)$ of a metric $g$ are natural and simple to compute. However, the boundary data  is not fully gauge-invariant or 
not fully geometric in the sense of \cite{I}, \cite{AA2}. In other words, consider the extension of the map $\Phi$ in \eqref{mainPhi} to general (non-gauged) 
vacuum Einstein metrics: 
\be \label{mainPhi*}
\begin{split}
&\hat \Phi: \bE \to [\cI_0 \times \bV'_S \times \cB]_c,\\
&\hat \Phi(g) = (g_S, K, \nu_S, [g_{\cC}], \mu_g).\\
\end{split}
\ee
In comparing \eqref{mainPhi*} and \eqref{mainPhiH}, in general, 
\be \label{notgeo} 
\hat \Phi \neq \Phi^H\circ \pi_*,
\ee 
where $\pi_*: \bE \to \cE_*$ is the projection to the orbit space and $\cE_*$ is identified with $\bE^H$ as in \eqref{harmslice}. The map $\hat \Phi$ is not 
constant on the fibers of the projection $\pi_*$; this is due to the fact that the groups $\Diff_1(M)$ and $\Diff_0(M)$ act non-trivially on the target data. To 
explain this in more detail, while the boundary data $[g_{\cC}]$ is invariant under the action of $\Diff_1(M)$ (and $\Diff_0(M)$), the volume term $\mu_g$ 
transforms as 
\be \label{mutrans}
(\f, \mu_g) \to \mu_{\f^*g} = (\f^*dv_g)(dv_{g_\cC})^{-2/n},
\ee
and generally $\f^* dv_g \neq dv_g$ along $\cC$. 

 For example, in $4$-dimensions, the boundary degrees of freedom for the metric $g$ require $10$ independent boundary conditions, $4$ of which are 
 associated with the choice of gauge. Here, of the remaining $6$ degrees of freedom, $5$ are geometric (corresponding to the specification of 
 $[g_{\cC}]$), but the scalar volume term $\mu_g$ is not (since it is not invariant under the action of $\Diff_0(M)$). This may be compared with the 
 works in \cite{FN},\cite{KRSW}. In \cite{FN} at most $3$ boundary conditions may be considered geometric in a restricted sense while the geometric 
 nature of the boundary conditions in \cite{KRSW} is less clear. Moreover, Theorem \ref{mainthm} holds in all dimensions while the works in \cite{FN} and 
 \cite{KRSW} are restricted to $4$-dimensions. 
 
  \medskip 
  
    Theorem \ref{mainthm} should be compared with two other recent related results. First, in \cite{I}, the IBVP was proved to be well-posed for Dirichlet boundary 
  data in an open neighborhood of solutions $(M, g)$ for which the Brown-York stress-energy tensor 
  \be \label{BY}
  \Pi = H g_{\cC} - A,
  \ee 
  has the same Lorentz signature $(- + \cdots +)$ as the induced metric $g_{\cC}$ on $\cC$. Here $A$ is the second fundamental form of $\cC \subset (M, g)$ 
  and $H = \tr_{g_{\cC}}A$ is the mean curvature of the boundary. While this is the first result on well-posedness for geometric boundary data, it does not hold in full generality, 
  i.e.~without some assumption such as \eqref{BY}, cf.~\cite{AA2}, \cite{AGM} for examples of ill-posedness and further discussion. Next, it was recently proved in \cite{LRSW} 
  that the IBVP for (the geometric) conformal mean curvature boundary data $([g_{\cC}], H)$ is ill-posed at the linearized level at certain backgrounds. More precisely, 
  this was proved to be the case when the IBVP is linearized at a standard spherical cylinder in Minkowski $\bR^{1,3}$, These remarks illustrate some the 
  difficulties in establishing a general well-posedness result for a geometric IBVP for the Einstein equations. Theorem \ref{mainthm} may be viewed as a 
  further step in this direction.  
  
  \medskip 
   
       Returning again to Theorem \ref{mainthm}, consider the map $\Phi^H$ in \eqref{mainPhiH} and use the identification \eqref{harmslice} to obtain via 
Theorem \ref{mainthm} the 
diffeomorphism
\be \label{Phi1}
\begin{split}
&\Phi_*: \cE_* = \bE/\Diff_1(M) \to [\cI_0 \times \bV'_S \times \cB]_c,\\
&\Phi_* [g] = (g_S, K, \nu_S, [g_{\cC}], \mu_g).\\
\end{split}
\ee
As noted above, the group $\Diff_0(M)$ acts on $\cE_*$ with quotient $\cE = \cE_*/\Diff_0(M)$ and with fibers given by the orbits of the action of 
$\Diff_0(M)/\Diff_1(M)$ on $\cE_*$. The group $\Diff_0(M)$ also acts on the target space in the standard way. The map $\Phi_*$ is 
equivariant with respect to these actions and so one obtains an induced map on the quotient spaces:
\be \label{Phi0}
 \begin{split}
 &\Phi_*: \cE = \bE/\Diff_0(M) \to \cI_0 \times_c [\cB],\\
&\Phi_* [g] = (g_S, K, [g_{\cC}]).\\
\end{split}
\ee
Here, the action of $\Diff_0(M)$ is trivial on the initial data $\cI_0 = \{(\g_S, \k)\}$ and is transitive on $\bV_S'$, so the quotient of the action on $\bV_S'$ is 
a single point, ignored in the notation above. Similarly, the action is trivial on ${\rm Conf}(\cC) = \{[\g_{\cC}]\}$ but as in \eqref{mutrans} is transitive on 
$C_+^{\infty}(\cC) = \{\mu\}$,\footnote{$\Diff_0(M)$ acts transitively on normal vectors $\nu_{\cC}$ to $\cC$ and thus on $\{\mu\}$.}  and so again 
the quotient of the action on $\{\mu\}$ is a single point. Thus $[\cB] = {\rm Conf}(\cC)$. This discussion immediately proves the following existence result. 

\begin{corollary} \label{surj}
The map $\Phi_*$ in \eqref{Phi0} is locally-in-time surjective. Thus, given arbitrary smooth initial data $(\g_S, \k) \in \cI_0$ together with a smoothly compatible 
conformal class $[\g_{\cC}] \in {\rm Conf}(\cC)$, there exist smooth vacuum Einstein metrics $(M, g)$ such that $\Phi_*(g) = (g_S, K, [g_{\cC}]) = (\g_S, \k, [\g_{\cC}])$. 
\end{corollary}

  However the map $\Phi_*$ is of course no longer injective. There is a full scalar degree of freedom in the domain $\cE$ mapping to the same point in 
target space. Theorem \ref{mainthm} shows that the fiber $(\Phi_*)^{-1}(\g_S, \k, [\g_{\cC}])$ in $\cE$ is naturally and smoothly identified with $\{\mu\} = C^{\infty}(\cC)$. 
In light of \cite{LRSW}, this is less clear in general when $\mu$ is replaced by the mean curvature $H$ of the boundary.  
 
 \medskip 
 
   Note that the boundary data $\cB$ is invariant under a restricted group of unimodular or volume-preserving diffeomorphisms. Thus, given 
 $g \in \bE^H$, let 
 $$\Diff_0^{dv_g} = \{\f \in \Diff_0(M): \f^*dv_g = dv_g \ {\rm at} \ \cC\},$$
 be the subgroup of diffeomorphisms in $\Diff_0(M)$ which fix the volume form of $g$ at $\cC$. Let $\cO_{\Diff_0^{dv_g}(M)}$ be the orbit of the 
 action of $\Diff_0^{dv_g}(M)$ on the space of metrics through the metric $g$. We may then form the bundle 
 $$\w \bE = \bigcup_{g \in \bE^H}\cO_{\Diff_0^{dv_g}(M)}.$$
 If $\w \pi: \w \bE \to \bE^H$ is the projection map, then one does have 
 $$\Phi | _{\w \bE} = \Phi^H\circ \w \pi: \w \bE \to \cI_0\times_c \cB.$$
 Thus $\Phi|_{\w \bE}$ is constant on the fibers of $\w \pi$ and is a diffeomorphism transverse to the fibers by Theorem \ref{mainthm}. 
 However, this does not hold on the full space $\bE$, as noted in \eqref{notgeo}. 
 
 \medskip 
 
    In general Theorem \ref{mainthm} implies two vacuum Einstein metrics $g_1, g_2 \in \bE$ are isometric (within the gauge group $\Diff_0(M)$) if and only 
 if there are diffeomorphisms $\f_i \in \Diff_0(M)$, $i = 1,2$, such that $\f_i^* g_i \in \bE^H$ and 
\be \label{PhiHuni}
\Phi^H(\f_1^*g_1) = \Phi^H(\f_2^*g_2).
\ee
This requires that $g_1$ and $g_2$ have the same initial data, induce the same conformal structure on the boundary $\cC$ and  
$$\f_1^* \mu_{g_1} = \f_2^*\mu_{g_2},$$
i.e.~
$$\f_1^*dv_{g_1} (dv_{(g_1)_{\cC}})^{-2/n} = \f_2^*dv_{g_2}(dv_{(g_2)_{\cC}})^{-2/n}.$$
We point out that this uniqueness for the gauged map $\Phi^H$ in \eqref{PhiHuni} cannot be improved to the ungauged map $\Phi$ in \eqref{mainPhi}.  
Thus it is not true that for general $g_1, g_2 \in \bE$ (not necessarily in harmonic gauge), 
\be \label{1-1}
\Phi(g_2) = \Phi(g_1) \Rightarrow g_2 = \f^*g_1,
\ee
for some $\f \in \Diff_0(M)$, i.e.~$\Phi$ in \eqref{mainPhi} cannot be one-to-one modulo isometry in $\Diff_0(M)$. To see this, suppose $g^1, g^2$ are 
a pair of vacuum Einstein metrics with the same initial data and the same Dirichlet boundary data so that $[g^1_{\cC}] = [g^2_{\cC}]$ and $dv_{g^1_{\cC}} = 
dv_{g^2_{\cC}}$.  As noted above, $\Diff_0(M)$ acts transitively on the space of bulk volume forms $dv_g$ restricted to $\cC$. Hence, there exists $\psi \in \Diff_0(M)$ 
such that for $\w g = \psi^*g^2$, $dv_{\w g} = \psi^*dv_{g^2} = dv_{g^1}$. Since $\Diff_0(M)$ leaves any volume form $dv_{g_{\cC}}$ invariant, it follows that 
$$\Phi(\w g) = \Phi(g^1).$$
Hence the validity of \eqref{1-1} would imply $\w g = \f^* g^1$ and so $g^2 = (\f \circ \psi^{-1})^*g^1$. Thus \eqref{1-1} implies uniqueness of solutions to the Dirichlet 
boundary value problem in $\cE$. However, this was proved to be false in general, (i.e.~without further conditions), in \cite{I}, \cite{AA2}, \cite{AGM}; it follows 
that \eqref{1-1} is also false without further conditions. \footnote{While it would be interesting, it is not clear if the issues discussed above can be overcome by 
passing to gravitational theories having a fixed volume form, such as unimodular gravity \cite{UG} or Weyl transverse gravity \cite{OR}.} 

\medskip

   The proof of Theorem \ref{mainthm} follows the general lines of the main results in \cite{I}. In particular, after a detailed analysis of the linearized problem, 
the non-linear result follows from the Nash-Moser inverse function theorem. 

\medskip 

  A brief summary of the contents of the paper is as follows. In \S 2, we introduce preliminary material needed for the work to follow. In \S 3, we derive the basic 
equations for the boundary data not part of the Dirichlet boundary data and prove the main local apriori energy estimates (Theorem \ref{Eestthm}) and the local 
existence theorem (Theorem \ref{modexist}), both at the linearized level. The work in \S 3 is then extended to the global setting of a general compact 
Cauchy surface in \S 4. The main (non-linear) results are then proved in \S 5 via the Nash-Moser theorem. 

\section{Initial Material}

  In this section, we discuss and derive background material needed for the main arguments to follow. The discussion is similar to that in 
\cite{I}, to which we occasionally refer for further details. 

 \subsection{Basic definitions and notation.}

  Let $S$ be a compact, connected and oriented $n$-manifold with boundary $\p S = \Si$; here the boundary $\Si$ may or may not be connected. Let $M = I\times S$, 
where $I = [0,1]$ is the range of a  time function $t: M \to I$. Let $S_t = \{t\}\times S$; $S_t$ is identified with $S_0$ upon introduction of local coordinates. 
We consider smooth Lorentz metrics $g$ on $M$ which are globally hyperbolic in the sense of manifolds with boundary, where the boundary $\cC = I\times \Si$ is 
timelike with respect to the metric $g$; the induced metric $g_{\cC}$ is thus a globally hyperbolic metric on $\cC$ with closed Cauchy surface $\Si$.\footnote{The 
results of this work thus apply to the `exterior' IBVP, where one of the boundary components of $\cC$ may be taken to infinity, keeping interior boundary 
components bounded.} 
 
   Let $\nu_S$ denote the future timelike unit normal vector field to $S \subset (M, g)$ and similarly let $\nu_{\cC}$ denote the (spacelike) 
 outward unit normal to $\cC \subset (M, g)$. Let 
 $$K = \tfrac{1}{2}\cL_{\nu_S}g |_S \ {\rm and} \  A = \tfrac{1}{2}\cL_{\nu_{\cC}}g |_{\cC},$$
denote the extrinsic curvatures (second fundamental forms) of $S$ and $\cC$ in $(M, g)$ respectively. 

   In local coordinates $x^{\a} = (t, x^a)$ near the corner $\Si$, $x^0 = t \geq 0$ is a defining function for $S = \{t = 0\}$, while $x^1 \leq 0$ is chosen 
to be a defining function for $\cC = \{x^1 = 0\}$. The coordinates $x^{\a}, \a = 2, \cdots, n$ are local coordinates of the corner $\Si$ itself. 
As usual, Greek letters $\a, \b$ denote spacetime indices $0, \dots, n$. Unless otherwise indicated, Roman letters denote spacelike indices $1, \dots, n$ 
and capital Roman letters denote corner indices $2, \dots, n$. 

\medskip 
 
  We will use the standard Sobolev norms to topologize the relevant function spaces. Let $D$ denote a domain contained either in  
$M$, $S$, $\cC$, $\Si$ or the corresponding $t$-level sets $S_t$, $\Si_t$, or the corresponding $t$-sublevel sets, $M_t$, $\cC_t$; $C_t = 
\{p \in \cC: t(p) \leq t\}$ and similarly for $M_t$.  Define the Sobolev $H^s$ norm on 
functions on $D$ by 
 $$||v||_{H^s(D)}^2 = \sum_{k=0}^s\int_{D} |\p_D^k v|^2dv_D,$$
where $\p_D^k$ consists of all partial coordinate derivatives tangent to $D$ of order $\leq s$. In place of coordinate derivatives, one may use the 
components of covariant derivatives of $v$ up to order $s$, given a background Riemannian metric $g_D$ on $D$. The volume form $dv_D$ is also 
induced from the metric $g_{D}$. We will also use the generalizations of the spaces $H^s$ for real-valued $s \in \bR^+$. 

  Define the stronger $\bar H^s$ norm by including all space-time derivatives up to order $s$, so that 
$$||v||_{\bar H^s(D)}^2 = \sum_{k=0}^s \int_{D} |\p_M^k v|^2dv_D,$$
where $\p_M$ denotes partial derivatives along all coordinates of $M$ at $D$. 
 
  Finally, for the Cauchy slices $S_t$, define the boundary stable $H^s$ norm on $S_t$ by 
$$||v||_{\bar \cH^s(S_t)}^2 = ||v||_{\bar H^s(S_t)}^2 + ||v||_{\bar H^s(\cC_t)}^2.$$
  
  We also recall the Sobolev trace theorem for the restriction map $f:D \to \bR$ to $f|_{\p D}: \p D \to \bR$:
$$||f||_{H^{s-1/2}(\p D)} \leq C||f||_{H^s(D)}.$$
It will always be assumed that $s$ is sufficiently large, depending on $n$, so that by the Sobolev embedding theorem, $H^s \subset C^2$.   

  We will use the function space  
\be \label{N}
\cN^s(M) = \cap_{j=0}^s C^j(I, \bar \cH^{s-j}(S)),
\ee
for functions on $M = M_1$. The space $\cN^s$ with associated norm $|| \cdot ||_{\cN^s}$ is a separable Banach space. The normed 
topology \eqref{N} is commonly used for the space of solutions of wave-type equations. 

  Note that 
$$H^{s+1}(M) \subset \cN^s(M) \subset H^s(M).$$

 Let $Met^s(M)$ be the space of globally hyperbolic Lorentz metrics $g$ on $M$, with coefficients $g_{\a\b} \in \cN^s$ in a smooth atlas for $M$. 
Similarly, let $(S^2(M))^{s-2}$ denote the space of symmetric bilinear forms on $M$ with coefficients in $H^{s-2}(M)$. 
 The initial data of $g$ is given by the pair $(g_S, K)$ consisting of the metric $g_S$ induced on the initial Cauchy surface $S = S_0$ and the 
 extrinsic curvature $K$ of $(S, g_S) \subset (M, g)$. This pair takes its values in the space $\cI^s = Met^s(S)\times (S^2(S))^{s-1}$ of all 
 initial data $(\g, \k)$, with the $H^s(S) \times H^{s-1}(S)$ topology.  
 
   For this work, the space of boundary data is 
 $$\cB^{s-1/2} = {\rm Conf}^{s-1/2}(\cC)\times H_+^{s-1/2}(\cC).$$
 The first factor consists of the space of pointwise conformal classes of globally hyperbolic Lorentz metrics on the boundary $\cC$, with 
 coefficients in $H^{s-1/2}(\cC)$. The second factor denotes positive $H^{s-1/2}$ functions on $\cC$. The bulk metric $g$ induces a 
 boundary metric $g_{\cC}$ and scalar $\mu_g = dv_g(dv_{g_{\cC}})^{-2/n}$ with 
 $$([g_{\cC}], \mu_g) \in \cB^{s-1/2}.$$ 
 Here $\mu_g$ is identified with a scalar function on $\cC$ by the choice of a fixed but otherwise arbitrary smooth volume form on $M$ and an 
 induced volume form on the boundary $\cC$. 

  We consider the map 
\be \label{Phi}
\Phi: Met(M) \to [S^2(M) \times \cI \times \cB]_c : = \cT,
\ee
$$\Phi(g) = (\Ric_g, g_S, K, [g_{\cC}], \mu_g),$$
where $(g_S, K)$ and $([g_{\cC}], \mu_g)$ are the initial and $\cB$-valued boundary data induced by $g$. The subscript $c$ denotes the compatibility 
conditions between the initial and boundary data at the corner $\Si$; these are discussed in detail in \S 2.3-2.4. For simplicity, we have dropped 
here the notation for the topologies, indexed by $s$. 

This is the extension of the map $\Phi$ in \eqref{mainPhi} to the space of all globally hyperbolic Lorentz metrics on $M$.

 \subsection{Localization} 
  
   As usual with hyperbolic PDE problems, the main existence results will first be proved locally, in regions where the constant coefficient 
approximation is effective. Such local solutions are then patched together, by a standard and well-known process, via rescaling and a 
partition of unity,

  Given a metric $g$ on $M$ as in \S 1, the localization at a point $p \in \Si$ is a (small) neighborhood $U \subset M$ of $p$, diffeomorphic via 
a local chart to a Minkowski corner 
$$\mathbf R=\{(t=x^0,x^1,\dots ,x^n):t\geq 0, x^1\leq 0\},$$ 
with $S\cap U \subset \{t = 0\}$, $\cC \cap U \subset \{x^1 = 0\}$ and $x^{\a}(p) = 0$. 

  In such a metrically small neighborhood $U$, the metric and coordinates are renormalized simultaneously by rescaling; thus 
for $\l$ small, set 
\be \label{rescale0}
\w g = \l^{-2}g, \ \ \w x^\a=\l^{-1}x^\a,
\ee
so that $\p_{\w x^{\a}} = \l \p_{x^{\a}}$ and  
\be\label{rescale}
\w g(\p_{\w x^\a},\p_{\w x^\b})|_{\w x} = g(\p_{ x^\a},\p_{ x^\b})|_{\l x}, 
\ee
and similarly for variations $h$, $\w h$ of $g$ and $\w g$. While the components $g_{\a\b}$ of $g$ are invariant under such a rescaling, 
all higher derivatives become small: 
\be \label{lambda}
\p_{\w x^{\mu}}^k \w g_{\a\b} = \l^k \p_{x^{\mu}}^k g_{\a\b} = O(\l^k).
\ee
Thus (for smooth metrics $g$) the coefficients are smoothly close to constant functions in the rescaled chart, 
\be \label{eps}
||\w g - g_{\a_0}||_{C^{k}(U)} \leq \e^k = \e^k(\l, g),
\ee
where $g_{\a_0}$ is a flat (constant coefficent) Minkowski metric. This is the frozen coefficient approximation. 

  The $\l \to 0$ limit blow-up metric $g_{\a_0}$ is the Minkowski metric, given in adapted coordinates as above by 
\be \label{Mincor}
g_{\a_0} =-dt^2-\a_0 dtdx^1+\sum_{i=1}^n(dx^i)^2.
\ee
Here $\a_0 = \a(p) = \<\nu_S, \nu_{\cC}\>_p$ is the value of the (scale-invariant) corner angle at $p$. 
Note that for $g_{\a_0}$,
\be \label{normals}
\nu_S = - \nabla t = (1+\a_0^2)^{-1/2}(\p_0+\a_0 \p_1),  \ \ \nu_\cC=\nabla x^1 = (1+\a_0^2)^{-1/2}(\p_1-\a_0 \p_0).
\ee 
Note also that the leading order Sobolev norms scale as follows: for $D \subset S$ or $D \subset \cC$ of dimension $n$, 
\be \label{Sobrescale}
||f||_{H_{\w g}^s(D)}^2 = \l^{2s-n} ||f||_{H_g^s(D)}^2,
\ee
$$||f||_{\cN_{\w g}^s(M)}^2 = \l^{2s-n} ||f||_{\cN_g^s(M)}^2 . 
\footnote{Here we are abusing notation and taking only the leading order $s$ derivatives of $f$; the scaling of the full Sobolev norm 
is anisotropic. Also, the range of the coordinates $\w x^{\a}$ and $x^{\a}$ has been suppressed in these equations but must suitably 
be taken into account.} $$
 
  \medskip 
 
   Also as in \cite{I}, $U$ is always assumed to be smoothly embedded in a larger region $\w U \subset M$,  
$$U \subset \w U,$$
with $\w U$ still covered by the adapted coordinates $(t, x^i)$, with $t \geq 0$ and $x^1 \leq 0$ in $\w U$ so that the initial surface $S$, boundary 
$\cC$ and corner $\Si$ in $\w U$ are an extension of the corresponding domains in $U$. We also assume \eqref{eps} still holds in $\w U$. 
All target data in $\cT$ as in \eqref{Phi} and later in $\cT^H$ which is given in $U$ is extended off $U$ to be of compact support in $\w U$ away 
from $S\cap U$ and $\cC \cap U$. Thus in such localizations, all target data vanishes in a neighborhood of the full timelike boundary of $\w U$ and in 
a neighborhood of the initial slice $\{t = 0\}$ away from $\cC \cap U$ and $S \cap U$ respectively. The same statements hold for variations of the 
target data, i.e.~in the tangent spaces $T(\cT)$ or $T(\cT^H)$. 

  From the geometric nature of the map $\Phi$ and target space $\cT$ in \eqref{Phi}, together with later variants below, it is easily verified that all the results 
  to follow are invariant under rescaling; they hold for $g$ and linearizations $h$ if and only if they hold for (all) localizations $\w g$ and $\w h$, 
  cf.~also \cite{I}. 

   For later reference, we note that solutions $h$ of linear systems of wave equations on $\w U$ appearing below have the finite propagation speed 
property; if the initial and boundary data of $h$ are compactly supported in $\w U$, then $h$ is also compactly supported in $\w U$, for some definite 
(possibly small) time $t > 0$. 

\subsection{Geometry at the corner}

  In this subsection, we discuss the geometry of the ambient metric $g$ at the corner $\Si$ and related compatibility conditions. First, let 
$\s$ denote a fixed background metric in the conformal class $[g_{\cC}]$, so that $g_{\cC} = \f^2 \s$, where $[\s]$ is determined by 
the target data but $\f$ is not. However, since the initial metric $g_S$ is target data, the induced metric $g_{\Si}$ on $\Si$ is also 
determined by target data. It follows that $\f |_{\Si}$ is determined by the initial and boundary data; one may set $\f = 1$ with suitable choice 
of $\s$. Consequently, the full metric $g_{\mu\nu}$ at $\Si$ is determined at the corner except possibly for the corner angle 
\be \label{angle}
\a = \<\nu_S, \nu_{\cC}\>.
\ee
In fact the next result shows that $\a$ is determined at $\Si$ by the initial and boundary data. 

\begin{lemma} \label{corner}
In local coordinates $\{x^{\mu}\}$ near $\Si$ with $g_{11} = f^2$ and $g_{1a} = 0$, $a = 2, \dots, n$ (zero shift), one has 
\be \label{dvg}
dv_g = \sqrt{f^2 det(-g_{\cC}) + \a^2 det(g_{\Si})} dx^0 \wedge \cdots \wedge x^n,
\ee
where $\a = -g_{01}$. Consequently, 
\be \label{mu1}
\mu_g = \frac{\sqrt{f^2 det(-g_{\cC}) + \a^2 det(g_{\Si})}}{[det(-g_{\cC})]^{1/n}},
\ee
and so at the corner $\Si$, $\a$ is uniquely determined, up to sign, by the initial and boundary data in $\cI \times \cB$. 

\end{lemma}

\begin{proof}
The formula \eqref{dvg} is a simple calculation in local coordinates, as is the relation $\a = -g_{01}$ and the formula \eqref{mu1}. As noted above, 
at $\Si$, both $f^2 = g_{11}$ and the conformal factor $\f$ are determined by initial data. Hence, $g_{\cC}$ and $g_{\Si}$ are determined at 
$\Si$ by the target data in $\cI \times \cB$. Thus the second statement follows from the formula \eqref{mu1}. 
\end{proof}

\begin{remark}
{\rm  A change $\a \to -\a$ of the sign of $\a$ corresponds to a change in time orientation $\p_t \to -\p_t$ and so this distinction will be 
ignored in the following. 
}
\end{remark}

\begin{remark} \label{vel}
{\rm The initial velocity of $\f$ at $\Si$ is also determined by the target data; in fact, one has the following expression  
\be \label{velocity}
(n-1)T_{\cC}(\f) = \sqrt{1+\a^2}\tr_\Si K_S - \a H_\Si - H_\s.
\ee 
where $T_{\cC}$ is the unit timelike normal to $\Si$ in $(\cC, g_{\cC})$ and $H_{\s}$ is the mean curvature of $\Si$ in $(\cC, \s)$. To prove 
\eqref{velocity}, note first that simple calculation shows that 
\be \label{TC}
T_{\cC} = \sqrt{1+\a^2} \nu_S - \a n_S,
\ee
where $n_S$ is the unit normal to $\Si$ in $(S, g_S)$. Taking the trace over $\Si$ of the covariant derivative of \eqref{TC} and applying a standard 
formula for the behavior of the mean curvature under conformal changes then gives \eqref{velocity}. Note that all the terms on the right side of \eqref{velocity} 
are determined by the initial and boundary data. The linearization of this result is used in Lemma \ref{u} below. 
}
\end{remark}

\subsection{Gauge analysis.}

  Due to gauge issues, the map $\Phi$ in \eqref{Phi} is well-known to be neither surjective nor injective, or given by non-linear hyperbolic evolution equations. 
A choice of gauge is used to deal with these issues. This will be described here by a choice of vector field $V$ on $M$ and replacing the vacuum Einstein equations 
\eqref{vacuum} by the reduced Einstein equations 
\be \label{Q}
 \Ric_g + \d^*V = 0,
\ee
where $\d^*V = \frac{1}{2}\cL_V g$. We choose here the (generalized) harmonic or wave coordinate gauge, given by 
\be \label{harm}
V_g = \Box_g x^{\a}\p_{x^{\a}},
\ee
where $\Box_g$ is the d'Alembertian or wave operator with respect to $g$. It is well-known that the gauge \eqref{harm} can be naturally globalized 
to all of $M$ by choice of a wave map to a fixed target space, cf.~\cite{I} for instance.  

  The gauged version of the map $\Phi$ in \eqref{Phi} is then given by 
\be \label{PhiH}
 \Phi^H: Met(M) \to \cT^H,
 \ee
$$\Phi^H(g) = \big((\Ric_g + \d_g^* V_g), (g_S, K, \nu_S, V|_S), ([g_{\cC}], \mu_g, V|_{\cC}) \big),$$ 
where $V|_S$ and $V|_{\cC}$ are the restrictions of $V_g$ to $S$ and $\cC$ respectively and $\nu_S$ is the future unit normal to $S$ in $(M, g)$. In 
$4$-dimensions, the initial data $(g_S, K, \nu_S, V|_S)$ consists of the (required) 20 degrees of freedom, while the boundary conditions 
$([g_{\cC}], \mu_g, V|_{\cC})$ consist of the (required) 10 degrees of freedom. Similar consistency holds in all dimensions. 

 The target space $\cT^H$ in \eqref{PhiH} is given as follows. Let $\cU$ denote the uncoupled space of target data, so that 
 $$\cU =  S^2(M) \times \cI \times \bV_S' \times \bV_S \times \cB \times \bV_{\cC}.$$
This is a product space, extending the product space in \eqref{Phi};  $\bV'_S$ is the space of $H^s$ vector fields along $S$ nowhere 
tangent to $S$, $\bV_S$ is the space of $H^{s-1}$ vector fields $V_S$ along $S$ and $\bV_{\cC}$ is the space of $H^{s-3/2}$ vector 
fields $V_{\cC}$ along $\cC$. By construction, $\cU$ is a separable Hilbert space, and Frechet space when $s = \infty$, i.e.~for $C^{\infty}$ 
data.  

  The metric $g$ induces data in $\cU$ through the map $\Phi^H$. The two equations for $g$, namely the bulk equation $\Ric_g + \d^*V_g = Q$ 
and the definition of the gauge field $V_g$ in \eqref{harm}, as well as definition of the normal vector $\nu_S$, determine compatibility 
conditions for the data in $\cU$ at the corner $\Si$. These compatibility conditions are expressed in terms of equations on the data in $\cU$. 
The subspace of compatible target data, denoted by   
$$\cU_c = [S^2(M) \times \cI \times \bV_S' \times \bV_S \times \cB \times \bV_{\cC}]_c \subset \cU, $$
is then the subspace satisfying these compatibility conditions. The compatibility conditions are expressed as the zero-locus of a (large) collection of 
functions on $\cU$, ordered according to the degree of differentiability.  

Note that the initial gauge velocity $\p_t V |_S$ is not prescribed by the data in $\cU$ or $\cU_c$. 
In fact, the gauge velocity at $S$ is only prescribed by imposing the constraint equations, i.e.~the Gauss and Gauss-Codazzi equations, 
on $S$: 
\be \label{Gauss}
|K|^2 - (\tr_{g_S}K)^2 - R_{g_S} = R_g - 2\Ric_g(\nu_S,\nu_S) = \tr_{g_S}\Ric_g - 3\Ric_g(\nu_S,\nu_S),
\ee
\be \label{GC} 
{\rm div}_{g_S}(K - Hg_S) = \Ric_g(\nu_S, \cdot). 
\ee 
To express this in terms of target data in $\cU_c$, denote
$$Q = \Ric + \d^*V,$$
so that \eqref{Gauss}-\eqref{GC} becomes 
\be \label{VGauss}
|K|^2 - (\tr_{g_S}K)^2 - R_{g_S} = \tr_{g_S}(Q - \d^*V) - 3(Q(\nu_S, \nu_S) - \d^*V(\nu_S,\nu_S)),
\ee
\be \label{VGC} 
{\rm div}_{g_S}(K - Hg_S) = Q(\nu_S, \cdot) - \d^*V(\nu_S, \cdot). 
\ee 
It is straightforward to check that $\d^*V$ is determined along $S$ by the initial data $(V|_S, \p_t V |_S)$. 
Let then $\bV_S$, the space of vector fields along $S$, representing the target space for $\p_t V |_S$ and consider the constraint operator 
 $$C: \cU_c \times \bV_S \to \bV_S,$$
\be \label{Const}
\begin{split}
& C(Q, \g_S, \k, V_S, \p_t V_S, \nu_S, \g_{\cC}, V_{\cC}) = 
\begin{cases}
|\k|^2 - (\tr_{\g_S}\k)^2  - R_{\g_S} - \tr_{g_S}(Q - \d^*V) + 3(Q(\nu_S, \nu_S) - \d^*V(\nu_S,\nu_S)),\\
{\rm div_{\g_S}}  \k - d_S(\tr_{\g_S}\k) - Q(\nu_S, \cdot) + \d^*V(\nu_S, \cdot). \\
\end{cases}
\end{split}
\ee
We then define 
 \be \label{TH}
 \cT^H = (\cU_c \times \bV_S)_0,
 \ee
to be the zero-set of the constraint operator $C$. By inspection in \eqref{Const}, on $\cT^H$, $\p_t V_S$ is uniquely determined by the remaining 
data in $\cU_c$, and so $\cT^H$ is naturally identified with $\cU_c$. 

The compatibility conditions on the initial and boundary data at the corner $\Si$ were analysed in detail in \cite{I} for Dirichlet boundary data; the 
analysis for the twisted Dirichlet (or conformal-volume) boundary data $([g_{\cC}], \mu_g)$ proceeds in exactly the same way with only minor changes 
and so we refer to \cite{I} details. In particular we note the following result, but refer to \cite{I} for the proof. 

\begin{proposition}\label{manifold}
With $C^{\infty}$ data, the target space $\cT^H$ in \eqref{TH} is a smooth tame Frechet manifold. 
\end{proposition}

\medskip 

  The overall method of proof of well-posedness is to prove that $\Phi^H$ in \eqref{PhiH} is a tame Fredholm map with a tame inverse at the 
linearized level and to then apply the Nash inverse function theorem in Frechet spaces. Thus, most of the analysis to follow deals with the 
study of the linearization 
\be \label{DF2}
D\Phi_g^H: T_g Met(M) \to T_{\tau}(\cT^H),
\ee
where $\tau = \Phi^H(g)$. 

\medskip 

  The bulk term in the linearization $D \Phi_g^H$ is given by 
 \be \label{L}
L(h) = \Ric'_h + \d^*V'_h + (\d^*)'_hV.
\ee
The linearization of the gauge field $V = V_g$ is given by 
\be \label{V'h}
V'_h = \b_g h - \<D^2 x^{\a}, h\>dx^{\a},
\ee
where $\b_g = \d_g + \frac{1}{2}d\tr_g h$ is the Bianchi operator. The bulk linearized equation $L(h) = F$ then has the form 
\be \label{L2}
L(h) = \tfrac{1}{2}D^*D h + P(h) = F,
\ee
where $P$ is a first order linear differential operator on $h$ and $D^*D = -\nabla_{e_{\a}}\nabla_{e_{\a}} + \nabla_{\nabla_{e_{\a}} e_{\a}}$ is the tensorial 
wave operator on symmetric bilinear forms. Of course all the coefficients of the derivatives in \eqref{L2} depend on the background metric $g$ at 
which the linearization is formed. 
 
  The equation \eqref{L2} may be written in local coordinates in the form 
 \be \label{L4}
(L(h))_{\a\b} = -\tfrac{1}{2} \Box_g (h_{\a\b}) + P_{\a\b}(h) = F_{\a\b}. 
\ee
where $P(h)$ involves $h$ and its first derivatives. In the localized context discussed in \S 2.2, the coefficients of $P$ are 
all of order $\e$, $P(h) = \e(\p h, h)$. Here $\Box_g f = \frac{1}{\sqrt{-g}}\p_{\a}(g^{\a\b}\sqrt{-g}\p_{\b}f)$ is the D'Alembertian or 
wave operator with respect to $g$. 

\medskip 

  In later sections, we will repeatedly need to solve linear wave equations of the form \eqref{L4}, i.e. 
\be \label{L0}
-\tfrac{1}{2}\Box_g v + p(v)  = \f,
\ee
for given $\f$. Here $v$ may be either scalar valued or vector-valued and again $p(v)$ is first order in $v$. Although solvability of the linear 
equation or system \eqref{L0} holds more generally, we will only need to actually solve such systems in $C^{\infty}$, so with $g$, $p$ and 
$\f$ in $C^{\infty}$. 

 The next Lemma is well-known and will be used repeatedly in \S 3. 

\begin{lemma} \label{Dir222}
For Dirichlet boundary data, the $C^{\infty}$ system \eqref{L0} is solvable for any smooth $\f$, with standard energy estimates. 
\end{lemma}

\begin{proof} 

  This is a standard and well-known result, cf.~\cite{BS}, \cite{Sa} for further details. Of course it is assumed here implicitly that the initial data 
and boundary data for $v$ satisfy the necessary smooth compatibility conditions (as above) at the corner $\Si$. Energy estimates for such 
equations are also standard; for later reference, we provide the main result here. Thus, the strong or boundary stable energy estimate with 
Dirichlet boundary data states that any (smooth) solution of \eqref{L0} satisfies the bound 
\be \label{DirE}
||v||_{\bar \cH^s(S_t)}^2 \leq C[||v||_{H^s(S_0)}^2 + ||\p_t v||_{H^{s-1}(S_0)}^2 + ||v||_{H^s(\cC_t)}^2 + ||\f||_{H^{s-1}(M_t)}^2],
\ee
where $t \in [0,1]$ and $C$ is a constant depending only on $g$ and the coefficients of $p$. Note the difference in the stronger and weaker 
norms on the left and right of \eqref{DirE}. We recall that $\cC_t = \{p \in \cC: t(p) \leq t\}$ and similarly for $M_t$. 

   The estimate \eqref{DirE} does hold for all $t \in [0, \infty)$, but then with a constant $C = C(t)$ which may grow (possibly exponentially) in time. 
This assumes of course that $g$ is defined on $[0, \infty)\times S$. 

\end{proof}

\begin{remark} \label{Neumann} 
{\rm One also has existence and uniqueness for smooth solutions of the equation \eqref{L0} with given Neumann boundary data $b(v) = \nu_{\cC}(v)$. 
However, in this case, there is no effective energy estimate as in \eqref{DirE}; instead the energy estimate only holds with a loss of a (fractional) derivative. 
The optimal result in general is: 
\be \label{NeuE}
|v||_{\bar \cH^s(S_t)}^2 \leq C[||v||_{H^s(S_0)}^2 + ||\p_t v||_{H^{s-1}(S_0)}^2 + ||\nu_{\cC}(v)||_{H^{s-1/3}(\cC_t)}^2 + ||\f||_{H^{s-1}(M_t)}^2],
\ee
so there is a 'loss' of $2/3$ of a derivative. We refer to \cite[Ch.8]{MT}, and in particular to \cite[Thm.9]{Tat} for the proof and further details. 
}
\end{remark}

 We next discuss the gauge field $V = V_g$ and its linearization $V'_h$ at $g$. Let 
\be \label{V0}
\Ric_g + \d^*V = Q,
\ee
so $Q$ is determined by the target data in $\cT^H$.  

\begin{lemma}\label{Gauge-lemma}
Let $\Box_g$ denote the wave operator $\Box_g = - D^*D$ acting on vector fields $V$ on $M$. Then  
\be \label{V}
-\tfrac{1}{2}[\Box_g + \Ric_g](V) = \b_g Q.
\ee
The initial data $V|_S$, $(\p_t V)|_S$ and boundary data $V|_{\cC}$ are determined by the target data  
$$(\Ric + \d^*V, (g_S, K_S, \nu_S, V|_S), V|_{\cC})$$
in $\cT^H$ and hence $V$ on $M$ is uniquely determined by target data. 

\end{lemma}

\begin{proof}
Applying the Bianchi operator $\b_g = \d_g + \frac{1}{2}d \tr_g$ to \eqref{V0} gives 
\be \label{V2}
\b_g \Ric_g + \b_g \d^*V = \b_g(Q),
\ee
which, by the Bianchi identity $\b_g \Ric_g = 0$, gives \eqref{V} via a standard Weitzenbock formula 

The Dirichlet data for $V$ along $S$ and $\cC$ are target data in $\cT^H$. We claim that the $t$-derivative $\p_t V$ on $S$ is also 
determined by target data. This follows from the Gauss and Gauss-Codazzi identities \eqref{Gauss}-\eqref{GC}. For $\nu = \nu_S$ the 
unit timelike normal to $S$, these identities show that the form $E(\nu, \cdot) = \Ric(\nu, \cdot) - \frac{1}{2}R g(\nu, \cdot)$ is determined by 
initial data $\iota = (g_S, K_S)$ along $S$. By \eqref{V0}, it follows that $\d^*V(\nu, \cdot) - \frac{1}{2}\tr \d^*V g(\nu, \cdot)$ is determined by 
initial data. Since $V$ is determined along $S$ by the target data, it follows easily that $\nabla_{\nu}V$ and hence $\p_t V$ is determined 
along $S$ by target data. It follows from the standard existence and uniqueness of solutions to the wave equation \eqref{V} with given initial 
and Dirichlet boundary data that $V$ is uniquely determined by target data. 

\end{proof}

 A similar result holds for the linearization $V'_h$ of the gauge field $V = V_g$. 
\begin{lemma}\label{Gauge-lemma2}
We have 
\be \label{V'}
-\tfrac{1}{2}(\Box_g + \Ric_g)V'_h + \tfrac{1}{2}\nabla_V V'_h = \b_g F + \b'_h \Ric_g + O_{2,1}(V,h),
\ee
where $O_{2,1}(V,h)$ is $2^{\rm nd}$ order in $V$ and $1^{\rm st}$ order in $h$. Moreover, $O_{2,1}(V,h) = 0$ if $V = 0$. 

  The initial data $(V'_h)|_S$, $(\p_t V'_h)|_S$ and boundary data $V'_h|_{\cC}$ are determined by the target data  
$$\big((\Ric + \d^*V)'_h, (h_S, K'_h, (\nu_S)'_h, V'_h), (V'_h)_{\cC}\big)$$
in $T(\cT^H)$.
\end{lemma}

\begin{proof}
Again applying the Bianchi operator to both sides of \eqref{L} gives:
\bes
\b_g \Ric'_h+\b_g \d^* V'_h+\b_g[(\d^*)'_h V]=\b_g F
\ees
which via the Weitzenbock formula as before implies 
\be \label{V'1}
-\tfrac{1}{2}[\Box V'_h+\Ric_g(V'_h)]=\b_g F + \b'_h \Ric_g -\b_g[(\d^*)'_h V].
\ee
Simple calculation gives $(\d^*)'_hV = \frac{1}{2}\nabla_V h + \d^*V\circ h$, so that 
 $$\b[(\d^*)'_hV] = \tfrac{1}{2}\b(\nabla_Vh) + O_{2,1}(V,h) = \tfrac{1}{2}\nabla_V \b h + O_{2,1}(V,h) = \tfrac{1}{2}\nabla_V V'_h + O_{2,1}(V,h),$$
where we have used \eqref{V'h} in the last equality. This gives \eqref{V'} 

  As in Lemma \ref{Gauge-lemma}, the Dirichlet data for $V'_h$ along $S$ and $\cC$ are given as target space data and we use the 
constraint equations \eqref{Gauss}-\eqref{GC} to determine the initial velocity $\p_t V'_h$. As before, the bulk equation yields:
\bes
\begin{split}
&(\Ric-\tfrac{1}{2}Rg)'_h(\nu_S)+[\d^*V-\tfrac{1}{2}({\rm div}V)g]'_h(\nu_S)\\
&=[F-\tfrac{1}{2}(\tr F)g](\nu_S)+\tfrac{1}{2}\<h,\Ric_g+\d^* V_g\>g(\nu_S)-
\tfrac{1}{2}\tr(\Ric_g+\d^* V_g)h(\nu_S)
\end{split}\ees
and thus
\bes\begin{split}
&[\d^*V-\tfrac{1}{2}({\rm div}V)g]'_h(\nu_S)\\
&=-[(\Ric-\tfrac{1}{2}Rg)(\nu_S)]'_h+(\Ric-\tfrac{1}{2}Rg)((\nu_S)'_h)\\
&+[F-\tfrac{1}{2}(\tr F)g](\nu_S)+\tfrac{1}{2}\<h,\Ric_g+\d^* V_g\>g(\nu_S)-
\tfrac{1}{2}\tr(\Ric_g+\d^* V_g)h(\nu_S)
\end{split}\ees
on $S$.
By the constraint equations \eqref{Gauss}-\eqref{GC}, $[(\Ric-\tfrac{1}{2}Rg)(\nu_S)]'_h$ is given by 
\bes
\big(-\tfrac{1}{2}[|K_S|^2-(\tr_{g_S}K_S)^2+R_S]'_{(h_S,K'_h)_S},~[{\rm div}_{g_S}K_S-d_S(\tr_{g_S}K_S)]'_{(h_S,K'_h)_S}\big).
\ees
Thus the target data in $T(\cT^H)$ uniquely determine the vector field 
\bes
[\d^*V'_h](\nu_S)-\tfrac{1}{2}({\rm div}V'_h)g(\nu_S),
\ees
along $S$. Since the initial data of $V'_h$ is already determined, this uniquely determines the vector field 
\be
\nabla_{\nu_S} V'_h \ \ {\rm along} \ \ S.
\ee

\end{proof}

\begin{remark} \label{gaugerem} 
{\bf (i).}
{\rm Lemma \ref{Gauge-lemma} implies the standard result that solutions of the gauge-reduced Einstein equations 
$$\Ric_g + \d^*V = 0,$$
with target data $ V = 0$ at $S\cup \cC$ and with $g$ satisfying the vacuum constraint equations \eqref{Gauss}-\eqref{GC} along $S$ necessarily satisfy 
$V = 0$ on $M$, and so are solutions of the Einstein equations 
$$\Ric_g = 0$$
on $M$. The same result holds locally, in the context of the localization in \S 2.2, as well as for the linearized equations \eqref{V'}. 

  Conversely, given any vacuum Einstein metric $g$, there is a unique diffeomorphism $\f \in \Diff_1(M)$ such that $\w g = \f^*g$ is in 
harmonic gauge, $V_{\w g} = 0$. Recall that $\Diff_1(M) \subset \Diff_0(M)$ is the group of smooth diffeomorphisms equal to the identity to 
first order on $S$. 
 
{\bf (ii).} The initial data $g_{\a\b}$, $\p_t g_{\a\b}$ for $g_{\a\b}$ along $S$ are determined by the initial data $g_S$, $K$, $\nu_S$ and $V|_S$ 
along $S$. To see this, observe that the determination of the corner angle $\a$ along $\Si$ (cf.~Lemma \ref{corner}) determines a choice of adapted 
local coordinates, up to the natural equivalence relation, cf.~\S 2.2. In such local coordinates, $g_S$ determines $g_{ij}$, $i, j = 1,\dots, n$. Also, since 
$\nu_S$ and the coordinate $t$ are given, the lapse-shift $(\mu, X)$ with $\nu_S = \mu \p_t + X$ are given, and so the components $g_{0\a}$ are also determined. 
Similarly, $K_S$ then determines $\p_t g_{ij}$. Finally, it is easy to see that from these prior determinations, $(V_S)^{\a} = \Box_g x^{\a}|_S$ determines 
$\p_t g_{0\a}$ on $S$. The same result holds for linearizations $h$ with the linearized target data. 

{\bf (iii).} For later use, note that the gauge variation $V'_h$ satisfies the boundary stable energy estimate \eqref{DirE}, since the Dirichlet boundary 
value $V'_h |_{\cC}$ is part of the target data in $T(\cT^H)$. 

}
\end{remark}

 \section{Linear Analysis I}
  
  In this section, we prove apriori energy estimates and the existence of $C^{\infty}$ solutions to the linearized problem \eqref{PhiH} for arbitrarily prescribed 
smooth target data in $T(\cT^H)$, in the localized setting of \S 2.2. The global version of the results here are proved in \S 4. 

  Throughout this section, the local rescaled metric $\w g$ will be denoted by $g$. Also, since most of the analysis concerns the boundary behavior, 
the unit normal $\nu_{\cC}$ to $\cC$ is generally denoted simply by $\nu$. It is  extended into bulk region $\w U$ as 
$$\nu_{\cC} = \frac{1}{|\nabla x^1|}\nabla x^1 =  \frac{1}{|\nabla x^1|}g^{\a1}\p_{\a}.$$
Clearly, $\nu = \nu_{\cC}$ is close to its Minkowski corner value, given by \eqref{normals}. 

\subsection{Boundary data equations.} 

  The boundary data in the target space $\cT^H$ consists of the data $([g_{\cC}], \mu_g, V_{\cC})$. The linearizations of these terms 
are given by 
\be \label{confb}
[g_{\cC}]'_h = \ring{h} = h^{\tT} - \frac{\tr_{\cC}h^{\tT}}{n}g_{\cC},
\ee
\be \label{mub}
(\mu_g)'_h = \tfrac{1}{2}(h_{\nu \nu} + \tfrac{n-2}{n}\tr_{\cC}h^{\tT})\mu_g,
\ee
\be \label{gauget}
(V'_h)^{\tT} = - \nabla_{\nu}h(\nu)^{\tT} + \d_{\cC}\ring{h}+ \tfrac{1}{2}d (h_{\nu \nu} + \tfrac{n-2}{n}u) - (A + H\g)h(\nu)^{\tT} - (\<D^2 x^{\a}, h\>\p_{\a})^{\tT},
\ee
\be \label{gaugen}
(V'_h)(\nu) = -\tfrac{1}{2}\nu(h_{\nu \nu}) + \tfrac{1}{2}\nu(\tr_{\cC}h^{\tT}) + \d_{\cC}(h(\nu)^{\tT}) - h_{\nu \nu}H + \<A, h^{\tT}\> - \<D^2 x^{\a}, h\>g_{\nu \a}. 
\ee 
Here the superscript $\tT$ denotes the component tangent to the boundary $\cC$. The derivations of \eqref{confb} and \eqref{mub} are elementary while 
the equations \eqref{gauget}-\eqref{gaugen} follow from the linearization of $V$ in \eqref{V'}. Namely, as noted in \cite{I}, straightforward 
computation from \eqref{V'h} gives 
$$(V'_h)^{\tT} = - \nabla_{\nu}h(\nu)^{\tT} + \b_{\cC}h{^\tT}+ \tfrac{1}{2}d h_{\nu \nu} - (A + H\g)h(\nu)^{\tT} - (\<D^2 x^{\a}, h\>\p_{\a})^{\tT},$$
Write $\b_{\cC}h^{\tT} = \d_{\cC}(\ring{h} + \frac{\tr_{\cC}h^{\tT}}{n}g_{\cC}) + \frac{1}{2}d\tr_{\cC}h^{\tT} = \d_{\cC}\ring{h} + \frac{n-2}{2n}d u$. 
Substituting this in the equation above gives \eqref{gauget}. The reason for the choice of the boundary term $\mu_g$ is exactly the appearance 
of $(\mu_g)'_h$ in \eqref{gauget}. Similar computation gives \eqref{gaugen}. 

For later purposes, we rewrite \eqref{gauget} as
\be \label{h_nuT} 
\nabla_{\nu}h(\nu)^{\tT} = Y(\cT) + E(h),
\ee
 where $Y(\cT) = -(V'_h)^{\tT} + \d_{\cC} \ring{h} + \frac{1}{2}d_{\cC}(h_{\nu\nu} + \frac{n-2}{n}\tr_{\cC}h^{\tT})$ is fully determined by the target data and 
 $E(h) = -(A + H\g)h(\nu)^{\tT} - (\<D^2x^{\a}, h\>\p_{\a})^{\tT}$ is regarded later as a small error term. 
 
 \medskip 
 
The target data $(\ring{h}, \mu'_h) \in T(\cT^H)$ represent Dirichlet boundary data for $h$ on $\cC$, while the terms $V'_h$ represent (to leading order) 
Neumann boundary data, together with given target data. Note that in these equations, there is no simple equation for the variation of the 
conformal factor 
$$u = \tr_{\cC}h^{\tT}.$$
In fact the next result shows that the linearization of the Hamiltonian constraint \eqref{Gauss} gives such an equation. 

\begin{lemma}\label{u}
The function $u$ satisfies the wave equation 
\be \label{uC}
-\tfrac{n-1}{n}\Box_\cC u-\tfrac{1}{n}R_\cC u = \w Y(\cT, \d^* V'_h) + \w E(h),
\ee
along the boundary $\cC$, where 
\be \label{uinhom}
\begin{split}
\w Y(\cT, \d^*V'_h)&=-R'_{\ring{h}} +\tr_g (F - \d^*V'_h) - 2 (F - \d^*V'_h)(\nu,\nu),\\
\w E(h)&= -2\<A'_h, A\> + 2H_{\cC}H'_h + 2\<A\circ A ,h\> - \tr_g (\d^*)'_h V + 2(\d^*)'_h V(\nu,\nu)\\
&\quad-4\Ric(\nu,\nu'_h) - \<\Ric, h\>.
\end{split}
\ee
Moreover, the initial data $u$ and $\p_t u$ for $u$ at the initial surface $\Si$ are determined by the initial and boundary data 
in $T(\cT^H)$. In particular, $\p_t u$ is determined at $\Si$ by the linearization of \eqref{velocity}. 

\end{lemma} 

\begin{proof} 
The linearization of the Gauss equation \eqref{Gauss} for the hypersurface $\cC$ gives:
\bes
2\<A'_h, A\>{ -2\<A\circ A, h\>}-2H_\cC H'_h + R'_{h^{\tT}}=R'_h-2\Ric'_h(\nu,\nu)-4\Ric(\nu,\nu'_h)
\ees
Note that 
\bes
R'_{h^{\tT}}=R'_{\tfrac{1}{n}ug_{\cC}}+R'_{\ring{h}}=-\Box_\cC u+\d_\cC\d_\cC(\tfrac{1}{n}ug_\cC)-\<\Ric_\cC,\tfrac{1}{n}ug_\cC\> + R'_{\ring{h}}=
-\tfrac{n-1}{n}\Box_\cC u-\tfrac{1}{n}R_\cC u+R'_{\ring{h}}
\ees
So we obtain:
\bes
-\tfrac{n-1}{n}\Box_\cC u-\tfrac{1}{n}R_\cC u=-R'_{\ring{h}}-2\<A'_h, A\>{ +2\<A\circ A,h\>}+2H_\cC H'_h+R'_h-2\Ric'_h(\nu,\nu)-4\Ric(\nu,\nu'_h). 
\ees
Also $R'_h=\tr \Ric'_h-\<h,\Ric\>$ and $\Ric'_h=F-\d^* V'_h-(\d^*)'_h V$. Thus
\bes\begin{split}
&-\tfrac{n-1}{n}\Box_\cC u-\tfrac{1}{n}R_\cC u\\
&=-R'_{\ring{h}} + \tr F -2F(\nu,\nu)\\
&\quad-2\<A'_h, A\> + 2H_{\cC}H'_h +2\<A\circ A,h\> - \tr [\d^*V'_h+(\d^*)'_h V] + 2[\d^*V'_h + (\d^*)'_h V](\nu,\nu) \\
&\quad - 4\Ric(\nu,\nu'_h) - \<h, \Ric\>,\\
\end{split}\ees
which proves the first statement. 

From the discussion in \S 2, cf.~Remark \ref{vel}, note $u$ and $\p_t u$ at $\Si$ are determined by the target data 
$h_S, K'_h, \ring{h}$ in $T(\cT) \subset T(\cT^H)$. 

\end{proof}

\subsection{Energy estimates}

   In this subsection, we derive (apriori) energy estimates for smooth solutions of the linear system $L(h) = F$ with linearized boundary 
data $(\ring{h}, \mu_h')$, in the local setting of \S 2.2. Using this, energy estimates in the global setting are then 
given in \S 4. 

 Let 
\be \label{tau'1}
\tau' = (F, (\g_S', \k', \nu', V_S'), (\ring{\g_{\cC}'}, \mu', V_{\cC}')) 
\ee
be a general element in $T_{\tau}(\cT^H)$, $\tau = \Phi^H(g)$. Define a Sobolev $H^s$ norm on the target space data in $T_{\tau}(\cT^H)$ as follows: 
for $\tau'$ as in \eqref{tau'1},  
\be \label{tnorm}
\begin{split}
||\tau'||_{H^s(\cT^H)} = [ ||F||_{H^s(\w U)} & + [ ||\ring{\g_{\cC}'}||_{H^{s+1}(\cC)} + ||\mu'||_{H^{s+1}(\cC)} + ||V'_{\cC}||_{H^s(\cC)}] \\
& + [ ||\g'_S||_{H^{s+1}(S_0)} + ||\k'||_{H^{s}(S_0)} + ||\nu'||_{H^s(S_0)} + ||V'_S||_{H^s(S_0)} ].
\end{split}
\ee
 On the right in \eqref{tnorm}, the top line consists of bulk and boundary terms, while the bottom line consists of initial data terms. Note the shift in the 
derivative index $s$ on the target compared with the discussion in \eqref{Phi} or \eqref{PhiH}. Of course $S_0$ and $\cC$ are the corresponding 
local domains in $\w U$. As above, the local rescaled quantities $\w g$, $\w h$ are denoted by $g$, $h$. 
 
  The main result of this section is the following energy estimate. As noted in \S 2.2, the result is invariant under rescaling by $\l$ as in \eqref{rescale0}. 

\begin{theorem} \label{Eestthm} 
 For the linearization $D\Phi_g^H$ of $\Phi^H$ in \eqref{PhiH} with $C^{\infty}$ data localized on $\w U \supset U$ 
as in \S 2.2, one has the tame energy estimate 
\be \label{Eest}
||h||_{\cN^s(\w U)}  \leq C||D\Phi_g(h)||_{H^s( \cT^H)},
\ee
where $C$ is a constant depending smoothly only on $(\w U, g)$. In particular, 
$$Ker D\Phi_g^H = 0.$$
\end{theorem} 

 \begin{proof} 
 
  Throughout the following $h$ is a symmetric bilinear form, written in adapted local coordinates $x^{\a} = (t, x^i)$ as 
$$h = h_{\a\b}dx^{\a}dx^{\b}.$$
The full variation $h$ is decomposed into a collection of terms as follows. On the boundary $\cC$, $h^{\tT}$ has the form 
$$h^{\tT} = h_{ij}dx^i dx^j,$$
where $i, j = 0, 2, \dots, n$. Then $h^{\tT}$ is extended into the bulk $\w U$ to have the same form, i.e.~without any $dx^1\cdot dx^{\a}$ 
term. Similarly, $h(\nu, \cdot)$ is defined on $\w U$ by 
\be \label{hnudef}
h(\nu, \cdot) = \frac{1}{|\nabla x^1|}g^{1 \a}h(\p_{\a}, \cdot),
\ee
which is defined by the background metric $g$ on $\w U$. 

   We will work with separate equations for each of the terms 
\be \label{comp}
h^{\tT}, \ h_{\nu \nu}, \ h(\nu)^{\tT},
\ee
comprising $h$. The bulk equation \eqref{L4} then becomes the following collection of equations: 
\be \label{e1}
-\tfrac{1}{2}\Box_g h^{\tT}  + P^{\tT}(h) = F^{\tT},
\ee
\be \label{e2}
-\tfrac{1}{2}\Box_g h_{\nu \nu}  + P_{\nu \nu}(h) = F_{\nu \nu},
\ee
\be \label{e3}
-\tfrac{1}{2}\Box_g h(\nu)^{\tT}  + P(\nu)^{\tT}(h) = F(\nu)^{\tT},
\ee

  By Lemma \ref{Dir222}, solutions of these equations satisfy the boundary stable energy estimate \eqref{DirE} for Dirichlet boundary data. 
Here the coupling term $P(h)$ is moved to the right side of each equation and treated as an inhomogeneous bulk term. We assume 
that $g$ and all target data are $C^{\infty}$ and the target data have compact support in $\w U$, away from $\cC \cap U$, cf.~\S 2.2. 
  Let $\tau' = D\Phi^H_g(h)$.
  
  Recall that $h^{\tT} = \ring{h} + \frac{u}{n}g_{\cC}$ and that the Dirichlet boundary data for $\ring{h}$ is given target data. 
Next the evolution or wave equation \eqref{uC} for $u$ along $\cC$ gives the following energy estimate for $u$ along $\cC$:
\be \label{ee1}
\begin{split}
||u||_{\w H^s(\Si_t)} \leq  C & [ ||F||_{H^{s-1}(\cC_t)} + ||\ring{h}||_{H^{s+1}({\cC_t)}} + ||V'_h||_{\bar H^s(\cC_t)} + \e ||h||_{\bar H^s(\cC_t)}] \\
+ & \ C[||u||_{H^s(\Si_0)} + ||\p_t u||_{H^{s-1}(\Si_0)}].
\end{split} 
\ee 
Here $\w H^s(\Si_t)$ denotes the Sobolev norm with both $t$-derivatives as well as derivatives tangent to $\Si_t$, but no derivatives normal to $\cC$. 
The extra derivative required on $\ring{h}$ comes from the $(R_{\cC})'_{\ring{h}}$ term in the $\w Y(\cT)$ term in \eqref{uinhom}. The term $\w E(h)$ in 
\eqref{uinhom} is of the form $\e(\p h)$, giving rise to the $\e$ term in \eqref{ee1}. Note all derivatives, including normal derivatives along $\cC$ are 
included in the $\e$-term in \eqref{ee1}.

  The last line in \eqref{ee1} consists of initial data and is bounded by 
$$||u||_{H^s(\Si_0)} + ||\p_t u||_{H^{s-1}(\Si_0)} \leq C[||\g'_S||_{H^{s+1/2}(S_0)} + ||\k'||_{H^{s-1/2}(S_0)}],$$
where we have used the Sobolev trace theorem and Remark \ref{gaugerem} and Remark \ref{vel}. Similarly, 
$$||F||_{H^{s-1}(\cC)} \leq C ||F||_{H^{s-1/2}(\w U)}.$$ 

   Next from the gauge wave equation \eqref{V'} and the boundary stable energy estimate with Dirichlet boundary data \eqref{DirE} we obtain 
\be \label{VbarV}
\begin{split}
||V'_h||_{\bar H^s(\cC_t)} \leq C&[ ||F||_{H^s(\w U)} + ||V'_{\cC}||_{H^s(\cC_t)} + \e||h||_{H^s(\w U)} ]\\
+ &C[||V'_h||_{H^s(S_0)} +  ||\p_t V'_h||_{H^{s-1}(S_0)} ]. 
\end{split}
\ee
Note that this estimate requires an extra derivative on the bulk $F$ term, due to the $\b F$ term in \eqref{V'}. Also, while $V'_h |_{S_0} = V'_S$ 
is given target data, the term $\p_t V'_h | _{S_0}$ is not. To obtain a bound on this term, we use the linearized constraint equations \eqref{VGauss}-\eqref{VGC}; 
this gives 
$$||\p_t V'_h||_{H^{s-1}(S_0)} \leq C[||h_S||_{H^{s+1}(S_0)} + ||K'_h||_{H^s(S_0)} +||F||_{H^{s-1}(S_0)} + ||V'_h||_{H^s(S_0)} ] \leq C||\tau'||_{H^s(\cT^H)}.$$

Given this control on the Dirichlet boundary value of $h^{\tT}$, \eqref{e1} and \eqref{DirE} then gives the estimate 
\be \label{ee2}
\begin{split}
||h^{\tT}||_{\bar \cH^s(S_t)} & \leq  C[||h^{\tT}||_{H^s(\cC_t)} + \e ||h||_{H^s(\w U_t)} + ||\tau'||_{H^s(\cT^H)}] \\
& \leq C[\int_0^t [ ||\ring{h}||_{H^{s+1}({\cC_{t'})}} +||V'_{\cC}||_{H^s(\cC_{t'})}.+ \e ||h||_{\bar H^s(\cC_{t'})} ]dt' +  \e ||h||_{H^s(\w U_t)} + ||\tau'||_{H^s(\cT^H)}], \\
\end{split}
\ee
where the last inequality comes from integration of \eqref{ee1}-\eqref{VbarV} over $t' \in [0,t]$. 
Taking the maximum over $t \in I$ then gives 
\be \label{ee2.1}
\max_{t\in [0,1]} ||h^{\tT}||_{\bar \cH^s(S_t)} \leq C[ ||\tau'||_{H^s(\cT^H)} + \e ||h||_{\hat H^s(\w U)}],
\ee
where $||h||_{\hat H^s(\w U)} :=  ||h||_{H^s(\w U)} + ||h||_{\bar H^s(\cC)}$. 

Similarly, the Dirichlet energy estimate for $h_{\nu \nu}$ gives 
$$\max_{t\in [0,1]}||h_{\nu \nu}||_{\bar \cH^s(S_t)} \leq C[ ||h_{\nu \nu}||_{H^s({\cC)}} + \e ||h||_{H^s(\w U)} + ||\tau'||_{H^s(\cT^H)} ], $$
and so using the definition of $\mu'_h$, $h_{\nu\nu} = \mu'_h - \frac{n-2}{n}u$, together with \eqref{ee2.1} gives 
\be \label{ee4}
\max_{t\in [0,1]} ||h_{\nu \nu}||_{\bar \cH^s(S_t)} \leq C[ ||\tau'||_{H^s(\cT^H)} + \e  ||h||_{\hat H^s(\w U)}],
\ee
Finally for $h(\nu)^{\tT}$, from the Neumann boundary estimate \eqref{NeuE}, we have 
$$||h(\nu)^{\tT}||_{\bar \cH^s(S_t)} \leq C[||\nabla_{\nu}h(\nu)^{\tT}||_{H^{s-1/3}({\cC_t)}} + \e ||h||_{H^{s}(\w U_t)} + ||\tau'||_{H^s(\cT^H)}  ]. $$ 
Examination of the form of the Neumann boundary data in \eqref{gauget} gives 
\be \label{ee5}
\begin{split}
||h(\nu)^{\tT}||_{\bar \cH^s(S_t)}  \leq & C[  ||\ring{h}||_{H^{s+2/3}(\cC_t)} + ||\mu'_h||_{H^{s+2/3}(\cC_t)}  \\
& + ||V'_{\cC}||_{H^{s-1/3}(\cC)} + \e(||h||_{H^{s-1/3}(\cC_t)} + ||h||_{H^{s}(\w U_t)}) + ||\tau'||_{H^s(\cT^H)} ].
\end{split}
\ee

  \medskip 
 
  Summing up the estimates \eqref{ee2.1}-\eqref{ee5}, by choosing $\e$ sufficiently small depending on $g$, one may absorb 
the $\e ||h||_{\bar H^s(\w U)}$ terms on the right side of these equations into the total sum on the left, using the fact that 
$|h||_{\bar H^s(\w U)} \leq \max_{t\in [0,1]}|h||_{\bar \cH^s(S_t)}$, giving then 
\be \label{Etot}
\max_{t\in [0,1]} ||h||_{\bar \cH^s(S_t)} \leq  C ||\tau'||_{H^{s}(\cT^H)} .
\ee
The estimate \eqref{ee5} which involves $\mu'$ in a crucial way, is a key ingredient in establishing the bound \eqref{Etot}. 

 The bound \eqref{Etot} gives a bound on 
\be \label{Linf}
h \in L^{\infty}(I, \bar \cH^s(S)).
\ee
However, from the regularity theory for wave equations with Dirichlet boundary data as in Lemma \ref{Dir222}, by performing the computations 
above over arbitrarily small time intervals $I' = [t_0,t_1] \subset I$, it is easy to verify that this can be improved to a bound on 
$$h \in C^0(I, \bar \cH^s(S))$$
so that \eqref{Etot} gives in fact 
$$||h||_{C^0(I, \bar \cH^s(S))} \leq  C ||\tau'||_{H^{s}(\cT^H)}.$$ 
Finally, as is standard, (cf.~also \cite{Sa}), taking $\frac{d}{dt}$ of the main equation $L(h) = F$ and commuting derivatives leads easily in the 
same way to the bound 
\be \label{Etot2}
||h||_{\cN^s(\w U)} \leq  C ||\tau'||_{H^{s}(\cT^H)}. 
\ee
This completes the proof of Theorem \ref{Eestthm}. 

\end{proof} 

\begin{remark} \label{timerescale}
{\rm The energy estimate \eqref{Eest} holds for time (and space) intervals $t^*$ on the order $O(1)$ in the localized coordinates $\w x^{\a}$. The rescaling 
\eqref{rescale0} then shows that \eqref{Eest} holds for unscaled $g$, $h$ on the order of $t^* = O(\l)$ for the unscaled 
metric $g$ and linearizations $h$. Moreover, the constant $C$ changes to $C\l^{-1}$ under such rescaling back to $g$, $h$, due to the Sobolev rescaling 
property \eqref{Sobrescale} of the domain and target norms. The proof above shows that $\l \sim \e$ must be chosen small enough to allow 
for absorption of $\l$ terms in the estimates above. As noted above, the constant $C$ in \eqref{Eest} depends smoothly on $g$, (i.e.~$\w g$). 

}
\end{remark}

 \subsection{Construction of local linearized solutions.} 

In this subsection, we construct solutions $h$ to the linear problem $D\Phi_g^H(h) =  \tau'$ for arbitrary $\tau' \in T(\cT^H)$. 
The main result is the following surjectivity or existence theorem. Again, this result is independent of the scaling \eqref{rescale0}. 
  
\begin{theorem}\label{modexist}
There exists $\e_0 > 0$ such that, for any smooth metric $g$ on $\w U$ which is $C^{\infty}$ $\e$-close, $\e \leq \e_0$, to a standard Minkowski 
corner metric $g_{\a_0}$ as in \eqref{eps}, the linearization $D\Phi_g^H$ at $g$ satisfies: given any $C^{\infty}$ target data 
$\tau' \in T_{\tau}(\cT^H)$ on $\w U$ as in \eqref{tau'1}, there exists a variation $h \in T_g Met(\w U)$, such that 
\be \label{DwPhisol}
D \Phi_g^H(h) = \tau'.
\ee
Thus, the equation 
\be \label{LhF}
L(h) = F,
\ee
has a $C^{\infty}$ solution $h$ such that 
\be \label{IC1}
(g_S, K)'_h = (\g_S', \k'), \ \ (\nu_S)'_h = \nu', \ \ V'_h = V_S' \ \mbox{ on }S,
\ee
and
\be \label{BC1}
(([g_\cC])'_h, \mu'_h) = (\ring{\g'}, \mu'), \ \ V'_h = V_\cC' \ \mbox{ on }\cC.
\ee
Further the constructed solution $h$ depends smoothly on the data $\tau'$. 
\end{theorem}
  
 \begin{proof} 
The proof is via an essentially standard Picard-type iteration process. Thus, we construct a sequence of approximate solutions uniformly bounded in 
$H^s$ norm (or $\cN^s$ norm) for any $s$, which converges to an exact solution of \eqref{LhF}-\eqref{BC1}. As in \S 3.2, we work with the components 
$h^{\tT}, h_{\nu\nu}, h(\nu)^{\tT}$ and the corresponding equations \eqref{e1}-\eqref{e3}. As previously, the notation $\w g$, $\w h$ is simplified to 
$g, h$. 
    
\medskip

0. In the initial step, we solve the relevant equations with the coupling terms set to zero, i.e.~the $P$-terms in the bulk and the $E$-terms in \eqref{uC} and \eqref{h_nuT} 
on $\cC$ are set to zero. Thus, the solution is constructed essentially depending only on the given target data in $T(\cT^H)$. This is done step-by-step 
for each of the terms in \eqref{comp}, 

  To begin, let $V_0'$ be the solution to the modified gauge equation 
\be \label{modV'}
-\tfrac{1}{2}(\Box_g V'_0+ \Ric_g (V'_0)) + \tfrac{1}{2}\nabla_{V} V'_0 = \b F,
\ee
(so the `error' term $\b'_h \Ric + O_{2,1}(V,h)$ is dropped from \eqref{V'}) and with initial and boundary data given by the target data $V'_S$ and $V'_{\cC}$.  
By the energy estimate \eqref{DirE}, we have 
$$||V_0'||_{H^s(\w U)} \leq C||\tau'||_{H^s(\cT^H)}.$$
 Next, let $u_0$ be the solution to the modified boundary equation \eqref{uC}, i.e. 
\be \label{modu}
-\tfrac{n-1}{n}\Box_\cC u_0-\tfrac{1}{n}R_\cC u_0 = -R'_{\ring{\g_{\cC}'}} + \tr_g(F - \d^*V_0') - 2(F - \d^*V_0')(\nu,\nu),
\ee
with the initial data for $u_0$ determined by the target data $\tau'$. Again by the energy estimate \eqref{DirE}, one has 
$$||u_0||_{H^s(\cC)} \leq C||\tau'||_{H^s(\cT^H)}.$$
Using this boundary value for $u_0$, we then solve the modified bulk equation for $h^{\tT}$, i.e. define $h_0^{\tT}$ to be the solution to 
$$-\tfrac{1}{2}\Box_g h_0^{\tT}   = F^{\tT},$$
on $\w U$, with initial and boundary conditions given by the target data $\tau'$ and $u_0$ determined as above. 
Again as above, this gives  
$$||h_0^{\tT}||_{H^s(\w U)} \leq C||\tau'||_{H^s(\cT^H)}.$$
The target data $\mu'$ then determines the Dirichlet boundary data for $(h_0)_{\nu\nu} = \mu' - \frac{n-2}{n}u_0$ and we solve 
$$-\tfrac{1}{2}\Box_g (h_0)_{\nu \nu}  = F_{\nu \nu},$$
with this Dirichlet boundary value for $h_{\nu \nu}$ and with the given initial conditions from $\tau'$. Again 
$$||(h_0)_{\nu\nu}||_{H^s(\w U)} \leq C||\tau'||_{H^s(\cT^H)}.$$
Finally we use the modified tangential boundary gauge equation \eqref{gauget}, i.e.
\be \label{gauget2}
\nabla_{\nu}h(\nu)^{\tT} = -(V'_h)^{\tT}+ \d_{\cC}\ring{h}+ \tfrac{1}{2}d (h_{\nu \nu} + \tfrac{n-2}{n}u) ,
\ee
and define the Neumann boundary data for $h_0(\nu)^{\tT}$ by 
$$\nabla_{\nu}h_0(\nu)^{\tT} = -(V_0')^{\tT}+ \d_{\cC}\ring{\g_{\cC}'}+ \tfrac{1}{2}d ((h_0)_{\nu \nu} + \tfrac{n-2}{n}u_0).$$
Let then $h_0(\nu)^{\tT}$ be the solution to the equation 
$$-\tfrac{1}{2}\Box_g h_0(\nu)^{\tT}  = F(\nu)^{\tT},$$
in $\w U$ with this given Neumann boundary data and with initial data given by the target data in $T( \cT)$, By the Neumann boundary estimate 
\eqref{NeuE} and the definition of the norm on $\cT^H$, this gives 
$$||h_0(\nu)^{\tT}||_{H^s(\w U)} \leq C||\tau'||_{H^s(\cT^H)}.$$

   Summing up the solutions above gives a solution $h_0$ to 
$$-\tfrac{1}{2}\Box_g h_0 = F,$$
in $\w U$ with the prescribed initial and boundary conditions from $\tau'$. Further  
$$||h_0||_{H^s(\w U)} \leq C||\tau'||_{H^s(\cT^H)}.$$

  The next stages are error, so $O(\e)$ adjustments to $h_0$. The construction of the iterative procedure is similar to the above, but with all target data set to 
$0$ with only small error terms remaining. 

 1.  To begin, solve the modified gauge equation for $V_1'$, (compare with \eqref{modV'} and \eqref{V'}),  
$$-\tfrac{1}{2}(\Box_g V'_1+ \Ric_g (V'_1)) + \tfrac{1}{2}\nabla_{V} V'_1 = \b_{h_0}'\Ric + O_{2,1}(V, h_0)$$
and with zero initial and boundary data. Since the right hand side involves at most first derivatives of $h_0$ with $\e$-coefficients, the Dirichlet 
energy estimate \eqref{DirE} gives  
$$||V_1'||_{H^s(\w U)} \leq \e C||\tau'||_{H^s(\cT^H)}.$$

   Next as in \eqref{modu}, we solve for $u_1$ on $\cC$ by requiring 
$$-\tfrac{n-1}{n}\Box_\cC u_1 -\tfrac{1}{n}R_\cC u_1 = \w E(h_0, V_1'),$$
where by inspection 
\bes
\begin{split}
\w E(h_0, V_1') = & -\tr_g \d^*V'_1 + 2 \d^* V'_1(\nu, \nu) -2\<A'_{h_0}, A\> + 2 H_{\cC}H'_{h_0} + 2\<A\circ A, h_0\> \\ 
& - 4\Ric(\nu, \nu)'_{h_0}) - \<\Ric, h_0\>  - \tr (\d^*)'_{h_0}V + 2 (\d^*)'_{h_0}V](\nu,\nu).
\end{split}
\ees
This term contains at most first order derivatives of $V_1', h_0$ with small $\e$-coefficients and hence $||u_1||_{H^s(\cC)} \leq C\e ||\tau'||_{H^s(\cT^H)}$. 
As in Step 0, let then $h_1^{\tT}$ be the solution in $\w U$ to the equation 
$$-\tfrac{1}{2}\Box_g h_1^{\tT} = - P(h_0),$$
with the Dirichlet boundary value for $u_1$ determined as above and the other target data of $h_1^{\tT}$ set to zero. It follows then as above that 
$$||h_1^{\tT}||_{H^s(\w U)} \leq C\e ||\tau'||_{H^s(\cT^H)}.$$
Similarly, solve 
$$-\tfrac{1}{2}(\Box_g h_1)_{\nu \nu}  = -P(h_0)_{\nu \nu},$$
with zero initial data and with boundary data given by $(h_1)_{\nu \nu} = -\frac{n-2}{n}u_1$ on $\cC$. This gives again 
$$||(h_1)_{\nu \nu}||_{H^s(\w U)} \leq C\e ||\tau'||_{H^s(\cT^H)}.$$
Finally for $h_1(\nu)^{\tT}$, we use the Neumann boundary data 
$$\nabla_{\nu} h_1(\nu)^{\tT} = E(V_1, h_0),$$
where from \eqref{h_nuT}, 
$$E(V_1', h_0) = (V_1')^{\tT} - (A - Hg_{\cC})h_0(\nu)^{\tT} - (\<D^2x^{\a}, h_0)\p_{\a})^{\tT}.$$
This only involves $V_1'$ and $h_0$ to zero order. Since $V_1' = O(\e)$ in $H^s$ and the coefficients of $h_0$ are also $O(\e)$, it follows that 
$\nabla_{\nu} h_1(\nu)^{\tT} = O(\e)$ in $H^s(\cC)$. Solving then the Neumann boundary value problem 
$$-\tfrac{1}{2}\Box_g h_1(\nu)^{\tT} = - P(h_0)(\nu)^{\tT},$$
with this boundary data and with zero initial data then gives 
$$||(h_1(\nu)^{\tT}||_{H^s(\w U)} \leq C\e ||\tau'||_{H^s(\cT^H)}.$$
     
     Summing up the solutions above gives a solution $h_1$ to 
$$-\tfrac{1}{2}\Box_g h_1 = -P(h_0),$$
in $\w U$ with zero target data and 
$$||h_1||_{H^s(\w U)} \leq C \e ||\tau'||_{H^s(\cT^H)}.$$
         
2. Next, we construct solutions $V_2'$, $h_2$ by exactly the same process as above with $h_0$ replaced by $h_1$. 
As before, this gives $V_2'$ and $h_2$ solving 
$$-\tfrac{1}{2}\Box_g h_2  = -P(h_1),$$
in $\w U$, with zero target initial data and zero target boundary data.  It follows as above that  
$$||h_2||_{H^s(\w U)} \leq C \e^2 ||\tau'||_{H^s(\cT^H)}.$$

   Continuing inductively in this way, gives the sequence $\{h_m\}$, $m = 1,2, \cdots$, satisfying  
$$||h_m||_{H^s(\cC)} \leq C \e^m ||\tau'||_{H^s(\cT^H)},$$
and hence the sequence $h^n = \sum_0^n h_m$ converges (weakly) in $H^s(\w U)$ to a limit $h \in H^s(\w U)$. Since $s$ is arbitrary, $h \in C^{\infty}$. 
By construction 
$$-\tfrac{1}{2}\Box_g(h_0 + \sum_1^{\infty} h_i)  = F - P(h_0 + \sum_1^{\infty} h_i),$$
so that 
\be \label{Lbarh}
L(h) = F.
\ee
Hence $h$ solves \eqref{LhF}. The sequence of gauge fields $(V')^n  = V_0' + \sum_1^{n}V_m'$ also converges weakly in $H^s$ to a limit field $V'$ 
satisfying \eqref{V'} and the initial and boundary conditions $V_S', V_{\cC}'$. We claim that in fact 
$$V' = V'_h.$$ 
To see this, note that both $V'$ and $V'_h$ have the same initial values on $S$ and $(V')^{\tT} = (V'_h)^{\tT}$ on $\cC$. Regarding the 
normal component, first observe that by construction, the limit $f = \lim f_i$ satisfies \eqref{uC}, i.e. 
\be \label{Hlin2}
\begin{split} 
\w \Box_{\cC} f + q(\p f) & =  R'_{\ring{h}} - \tr(F - \d^* V') + 2(F - \d^* V')(\nu,\nu)\\
 &+2\<A'_h, A\> - 2H_\cC H'_h - 2\<A\circ A, h\> + 4\Ric(\nu,\nu'_h) + \<\Ric, h\>\\
&+ \tr [(\d^*)'_h V] - 2[(\d^*)'_h V](\nu,\nu), \\
\end{split}
\ee
along $\cC$. Moreover, by \eqref{Lbarh}, $(\Ric + \d^*V)'_h = F$ and so applying the linearization of the Hamiltonian constraint as before shows 
that \eqref{Hlin2} also holds with $V'_h$ in place of $V'$. Hence  
$$\tr(\d^*V')-2\d^* V'(\nu,\nu) = \tr(\d^* V'_h) - 2\d^*  V'_h(\nu,\nu) \mbox{ on }\cC.$$
Since $(V')^{\tT} = (V'_h)^{\tT}$, this equation shows that 
$$\nu \<V', \nu\> = \nu\<V'_h, \nu\> \ \ {\rm at} \ \ \cC.$$
Since both $V'$ and $V'_h$ satisfy the same gauge equation \eqref{V'}, it then follows from uniqueness of solutions to the mixed Dirichelt-Neumann 
boundary value problem that $V' = V'_h$. 

Also, the construction is clearly smooth in the data $T(\w \cT^H)$ and in $g$. This completes the proof.

\end{proof}

 \section{Linear Analysis II}

 In this section, we prove the global-in-space version of the results of \S 3. This is done by patching together local solutions and estimates to obtain 
 global-in-space solutions and estimates. The following result proves the well-posedness of the IBVP for the linearized problem, allowing for loss-of-derivatives.

\begin{theorem}\label{globexist}
Let $(M, g)$ be a smooth globally hyperbolic Lorentzian manifold, $M \simeq I \times S$, with compact Cauchy surface $S$ having compact boundary 
$\p S = \Si$ and with timelike boundary $(\cC, g_{\cC})$. Then for any smooth 
$\tau' \in T_{\tau}(\cT^H)$, the equation 
\be \label{L3}
D\Phi_g^H(h) = \tau'
\ee
has a unique $C^{\infty}$ solution $h$ defined on the domain $[0, t^*)\times S$ where $t^* = O(\l)$, $\l = \l(g)$. The solution $h$ depends 
smoothly on the data $\tau'$. Moreover, any $h$ satisfying \eqref{L3} satisfies the tame energy estimate 
\be \label{Eest2}
||h||_{\cN^s(I^*\times S)}  \leq C\l^{-1}||D\Phi_g^H(h)||_{H^s(T^H)}, 
\ee
where $I^* = [0, t^*)$, cf.~Remark \ref{timerescale} and $C$ depends smoothly only on $g$. 
\end{theorem} 
 
 \begin{proof}

  Given the previous local results, the proof of Theorem \ref{globexist} is essentially standard. We provide the details below, 
 following the proof of the analogous result in \cite{I}. 
  
  Given $(M, g)$ as in Theorem \ref{globexist}, first choose an open cover $\{U_i\}_{i=1}^N$ of the corner $\Si$, where each $U_i$ is sufficiently 
small so that the local existence result Theorem \ref{modexist} holds for each $U_i$. Also, choose an open set $U_0$ in the interior $M\setminus \cC$ 
such that $(\cup_{i=1}^N U_i)\cup U_0$ also covers a tubular neighborhood of the initial surface, i.e. $S\times [0,t_0]$ for some time $t_0>0$.

Let $\{\w U_i\}_{i=0}^N$ be a thickening of the cover $\{U_i\}_{i=0}^N$, so that $U_i \subset \w U_i$, as described in \S 2.2. Let 
$\rho_i$ be a partition of unity subordinate to the cover $\{\w U_i\}_{i=0}^N$, so that ${\rm supp}\rho_i\subset \w U_i$, $\rho_i = 1$ on 
$U_i$ and  $\sum_i \rho_i=1$. Let 
$$\tau' =\big(F, (\g', \k', \nu', V_S'), (\ring{\g_{\cC}'}, \eta',V_{\cC}') \big),$$
be arbitrary $C^{\infty}$ data in $T(\cT^H)$ on $M$. Then the data $\rho_i \tau'$ has compact support in the sense of \S 2.2 in $\w U_i$ 
and by Theorem \ref{modexist} there exists a solution $h_i$ in $\w U_i$ satisfying 
\be \label{hitau}
D\Phi^H_g(h_i) = \rho_i \tau'. 
\ee
(Here $\tau'$ is rescaled by the relevant powers of $\l$ as in \S 2.2, but the notation will not reflect this). Moreover, as noted at the end of \S 2.2, by the finite propagation 
speed property, the solution $h_i$ has compact support in the sense of \S 2.2 for $t \in [0,t_i]$ for some $t_i > 0$. Thus, $h_i$ extends smoothly as the 
zero solution on $M_{t_i} \setminus \w U_i$. 

   Similarly, in the interior region $U_0 \subset \w U_0$, let $h_0$ be the solution to the Cauchy problem 
$$L(h_0) = \rho_0 F \ \mbox{ in } \w U_0, \ \ (g_S, K, \nu_S, V_g)'_h = \rho_0(\g', \k', \nu', V_S') \ \mbox{ on }S_0.$$
Then again $h_0$ has compact support in $\w U_0$ for some time interval $[0, t_0]$ and as above extends to $M_{t_0}$. We relabel so that 
$t_0$ is a common time interval for all $h_i$, $i \geq 0$. 

  The sum 
\be \label{sum}
h = \sum_{i=0}^N h_i,
\ee
is thus well-defined on $M_{t_0}$ and by linearity
\be \label{htau}
D \Phi_g^H(h) = \tau'.
\ee
(For this global solution, the rescaling above is of course removed). By the local energy estimate in Theorem \ref{Eestthm}, the solution $h_i$ of \eqref{hitau} is 
unique. It follows that the global solution $h$ is the unique smooth solution \eqref{htau}, since if $h'$ is any other solution of \eqref{htau}, one may apply the 
local decomposition \eqref{sum} in the same way to $h'$. Moreover, the partition of unity property shows that the local energy estimates from \eqref{Eest} sum 
to give the global energy estimate \eqref{Eest2}. The factor of $\l^{-1}$ in \eqref{Eest2} arises from the different scaling behavior of the Sobolev norms in the domain and 
target spaces. 

  The smoothness statement follows from that in Theorem \ref{modexist}. 

\end{proof}

\section{Nash-Moser theory}
   
   The analysis in \S 4 shows that the linearized IBVP with twisted Dirichlet boundary data is locally-in-time well-posed and has tame energy 
estimates if one allows for a loss of derivatives for the initial and boundary data. This is the setting of the Nash-Moser implicit function theorem, 
which we now apply to prove Theorem \ref{mainthm}. 

\medskip 

 First define the space $Met^*(M)$ to be the space of globally hyperbolic metrics $g$ (in the sense of manifolds with timelike boundary) defined (at 
 least) on a domain of the form $[0, t^*)\times S$, where, as in \eqref{tstar} 
 $$t^*: Met(M) \to \bR,$$
 is a smooth function of $g$. The value $t^*(g)$ represents the $g$-proper time between the slices $\{0\} \times S$ and $\{t^*\}\times S$. Define the 
 target data space $(\cT^H)^*$ in the same way, so that the data on $\cC$ are restricted to $[0, t^*)$. 
  
  We then have the following general off-shell result. 

\begin{theorem} \label{offshell}
There exists a smooth choice of the time function $t^*$ for $Met^*(M)$ such that the mapping 
\be \label{offPhi}
\Phi^H: Met^*(M) \to (\cT^H)^*
\ee
$$\Phi^H(g) = \big(\Ric_g + \d^*V, (\g_S, K, \nu_S, V|_S), ([g_{\cC}], \mu_g, V|_{\cC}) \big),$$
is a smooth tame diffeomorphism between tame Frechet manifolds. 
\end{theorem}

\begin{proof} 
By Proposition \ref{manifold}, the target space $\cT^H$ is a smooth tame Frechet manifold, as is the domain $Met(M)$; the same holds for $(\cT^H)^*$ and 
$Met^*(M)$. By construction, it is easy to see that the map $\Phi^H$ in \eqref{offPhi} is then a smooth, tame map between Frechet manifolds, cf.~\cite{Ham} 
for the basic definitions regarding tame maps and Frechet manifolds here. 

  We choose the time function $t^*$ (for instance) as follows. Let $t^*(g) = \l(g)$, where $\l$ is chosen smoothly so that the maximum of the $\e = \e(\l)$ factors in 
\eqref{ee2.1}-\eqref{ee5} satisfies $\e C \in [\frac{1}{4}, \frac{1}{2}]$. With this choice, by Theorem \ref{globexist}, the derivative $D\Phi_g^H$ is then a tame 
isomorphism at any $g \in Met^*(M)$ and by the energy estimate \eqref{Eest2} the inverse $(D\Phi_g^H)^{-1}$ at $\tau(g) \in (\cT^H)^*$ exists and is a tame 
linear map of Frechet spaces. The Nash-Moser inverse function theorem, cf.~\cite{Ham}, \cite{Z2}, then implies that $\Phi^H$ is everywhere a smooth 
local diffeomorphism. 

  Next we claim that $\Phi^H$ in \eqref{offPhi} is not only locally one-to-one but in fact globally one-to-one. To see this, suppose $g_1, g_2$ are two metrics 
with the same gauged target data, i.e.~$\Phi^H(g_1) = \Phi^H(g_2)$. If $g_1 \neq g_2$, then there exists an arbitrarily small neighborhood 
$U$ of some $p \in \Si$ such that $g_1 \neq g_2$ in $U$. One may then perform a localization or rescaling of each $g_i$ as in \S 2.2 so that the almost 
constant coefficient metrics $\w g_i$ satisfy $\w g_1 \neq \w g_2$ in $U$. However, by \eqref{eps}, the metrics $\w g_i$ satisfy 
$$|| \w g_2 - \w g_1||_{C^{\infty}} \leq \e,$$
in $U$, for a prescribed $\e > 0$. It then follows from the local uniqueness in the Nash-Moser theorem above that 
$$\w g_2 = \w g_1$$
in some $U' \subset U$. Hence $g_2 = g_1$ in $U'$, which proves the global uniqueness. 

  To see that $\Phi^H$ is globally surjective, let $\tau \in \cT^H$ be arbitrary. Consider a small corner neighborhood $U$ of $p \in \Si$ and rescale the data 
$\tau$ to $\w \tau$. Then $Q \sim 0$, $\g_S$ is close to a flat half-space $\{x^1 \leq 0\}$ in $\bR^n$, $\k \sim 0$, $\nu_S$ is close to the value in \eqref{normals} 
and $V \sim 0$. Thus $\w \tau$ is close to the value $\w \tau_0 =  \Phi(\w g_{\a_0})$, where $\w g_{\a_0}$ is the Minkowski corner metric in \eqref{Mincor}. 
The Nash-Moser theorem above applies and gives the existence of $\w g$ in $U$ such that $\Phi(\w g) = \w \tau$. Rescaling back then gives the existence 
of $g$ such that $\Phi(g) = \tau$ in $U$. Since $p \in \Si$ is arbitrary, such local solutions exist everywhere near $\Si$ and the uniqueness above implies 
the local solutions agree on overlaps. Together with existence and uniqueness of solutions to the Cauchy problem on $S$, this gives the existence of $g$ 
on $M$ (for small time) with $\Phi(g) = \tau$. 
  
This completes the proof of Theorem \ref{offshell}. 

\end{proof}

Next we turn to the proof of Theorem \ref{mainthm}

\begin{proof}
Let 
$$\cT^H_0 = \{ \big(0, (\g, \k, \nu_S, 0), (\g_{\cC}, \mu, 0)\big) \} \subset \cT^H,$$
so $Q$ and the initial and boundary data of the gauge field in $\cT^H$ are set to zero. It is proved in \cite{I} that $\cT^H_0$ is a smooth submanifold 
of $\cT^H$. (The fact that the boundary data here is twisted Dirichlet boundary data as opposed to Dirichlet boundary data plays no role in the proof). 
Throughout the following, we work in the local-in-time context, but for simplicity drop the $^*$ from the notation in domain and target spaces. 

  Observe that 
$$(\Phi^H)^{-1}(\cT^H_0) = \bE^H \simeq \cE_*.$$
Namely, $\Phi^{-1}(\cT^H_0)$ is the space of Lorentz metrics $g$ on $M$ with $Q  = 0$ on $M$ and $V = 0$ on $S \cup \cC$. Since then 
$\Ric_g + \d^*V = 0$, it follows from Remark \ref{gaugerem} that $V = 0$ on $M$. Thus the metric $g$ is in harmonic gauge and $\Ric_g = 0$, 
with given unit timelike normal $\nu_S$ at $S$. 
It follows then from Theorem \ref{offshell} that the map 
$$\Phi^H: \bE^H \to \cT^H_0,$$
is a smooth, tame diffeomorphism between Frechet manifolds, with smooth, tame inverse $(\Phi^H)^{-1}$. This completes the proof of Theorem \ref{mainthm}. 

\end{proof} 

Next we give the proof of Corollary \ref{Eman}. 

\begin{proof}
We have shown in the proof of Theorem \ref{mainthm} above that $\bE^H$ is a smooth, tame Frechet manifold. The group $\Diff_0(M)/\Diff_1(M)$ is a 
smooth tame Frechet Lie group (cf.~\cite[I.4.6]{Ham}) which acts smoothly and freely on $\bE^H$. It then  follows from standard theory, 
cf.~\cite[Cor.III.2.5.2]{Ham} that the quotient $\bE^H / (\Diff_0(M)/\Diff_1(M))$ is a smooth, tame Frechet manifold. As noted in \S 1, 
$$\cE \simeq \bE^H / (\Diff_0(M)/\Diff_1(M)),$$
and hence $\cE$ is a smooth, tame Frechet manifold. 

\end{proof}

\begin{remark} \label{findiff} 
{\rm All of the discussion above has been phrased in the $C^{\infty}$ context. However, all the results may be extended to Sobolev spaces $H^s$ of finite 
differentiability. Thus, as discussed in \S 2, the maps $\Phi$ and $\Phi^H$ are well-defined smooth maps 
between a scale of Banach manifolds, 
$$\Phi^H: Met^s(M) \to (\cT^H)^s,$$
indexed by $s$. The energy estimates in \eqref{Eest2} hold on the $H^s$ completions of the domain and target spaces for suitable values of $s$, giving the 
existence of a tame inverse for the derivative $D\Phi^H$. One may then apply the finite derivative version of the Nash-Moser inverse 
function theorem due to Zehnder \cite{Z1}, or \cite[Thm.6.3.3]{Z2}; in the notation used there, the loss of derivative is $\g = 3$. 
However, we will not carry out here the details that the hypotheses of \cite[Thm.6.3.3]{Z2} do hold. 

}
\end{remark}
 
 \begin{remark}
 {\rm 
  As in \cite{I}, the results in this work hold for any value of the cosmological constant $\Lambda \in \bR$. The results also hold for the Einstein equations 
coupled to Kaluza-Klein matter fields, which derive from the vacuum Einstein equations in higher dimensions by symmetry reduction. The boundary data 
of the KK fields are then induced from the higher dimensional vacuum boundary data in $\cB$. 
}
\end{remark}

\bibliographystyle{plain}

\end{document}